\documentclass[journal]{IEEEtran}

\usepackage{cite}
\usepackage{amsmath,amssymb,amsfonts}
\usepackage{algorithmic}
\usepackage[ruled,lined]{algorithm2e}
\usepackage{graphicx}
\usepackage{textcomp}
\usepackage{xcolor}
\usepackage{mathtools}
\usepackage{siunitx}
\usepackage{bm}
\usepackage{braket}
\DeclareMathOperator{\Tr}{\mathrm{Tr}}
\newtheorem{theorem}{Theorem}

\newtheorem{lemma}[theorem]{Lemma}
\newtheorem{remark}{Remark}
\newtheorem{corollary}[theorem]{Corollary}

\newtheorem{proposition}[theorem]{Proposition}

\newcommand{\supp}{\operatorname{supp}}

\newcommand{\argmin}{\mathop{\rm argmin}\limits}

\newcommand{\dist}{\operatorname{dist}}

\newcommand{\nc}{\newcommand}
\nc{\ketbra}[2]{|#1\rangle\!\langle#2|}
\nc{\st}{{\text{ s.t. }}}
\nc{\tr}{\operatorname{Tr}}

\def\BibTeX{{\rm B\kern-.05em{\sc i\kern-.025em b}\kern-.08em
    T\kern-.1667em\lower.7ex\hbox{E}\kern-.125emX}}

\nc{\cM}{{\cal M}}
\nc{\cN}{{\cal N}}
\nc{\cS}{{\cal S}}
\nc{\cH}{{\cal H}}
\nc{\cG}{{\cal G}}
\nc{\cF}{{\cal F}}
\nc{\cI}{{\cal I}}
\nc{\cD}{{\cal D}}
\nc{\cB}{{\cal B}}
\nc{\cT}{{\cal T}}
\nc{\cL}{{\cal L}}

\def\Label#1{\label{#1}\ [\ \text{#1}\ ]\ }
\def\Label{\label}

\begin{document}
\title{Conditions for Global Optimality in Quantum Arimoto–Blahut Algorithms}
\author{Geng~Liu and Masahito~Hayashi%
\thanks{M. Hayashi was supported in part by the Guangdong Provincial Quantum Science Strategic Initiative under Grant GDZX2505003, the General R\&D Projects of the 1+1+1 CUHK--CUHK(SZ)--GDST Joint Collaboration Fund under Grant GRDP2025-022.}%
\thanks{Geng Liu is with the School of Science and Engineering, The Chinese University of Hong Kong, Shenzhen, China, and with the International Quantum Academy (SIQA), Shenzhen, China.}%
\thanks{Masahito Hayashi is with the School of Data Science, The Chinese University of Hong Kong, Shenzhen, China, with the International Quantum Academy (SIQA), Shenzhen, China, and with the Graduate School of Mathematics, Nagoya University, Nagoya, Japan (e-mail: hmasahito@cuhk.edu.cn).}}

\maketitle
\begin{abstract}
Generalized Arimoto--Blahut (AB) algorithms are widely used in information theory and quantum optimization, but monotonic objective decrease and numerical stabilization do not guarantee global optimality. 
We establish necessary and sufficient conditions for the global optimality of full-rank AB fixed points for convex differentiable objectives under linear constraints. 
An AB fixed point is globally optimal if and only if the AB update direction and the objective gradient differ by an element of the constraint normal space at that point. The same compatibility condition characterizes agreement between individual AB and mirror-descent (MD) updates, whereas pathwise equivalence requires it along the entire common trajectory. 
Thus, an AB algorithm may follow a trajectory different from MD and still reach the global optimum. We also develop a posteriori optimality certificates based on feasible directional derivatives, including finite-difference upper bounds on the objective gap that require only objective evaluations. 
For channel relative entropy between dephasing and depolarizing channels, we analytically identify the global minimizer as the unique full-rank AB fixed point, although pathwise equivalence fails. 
Numerical experiments illustrate convergence of AB and MD along different trajectories and validate the certificates. By contrast, an amplitude-damping example shows that a monotone AB iteration can stabilize at a suboptimal fixed point, whose nonoptimality is detected by the finite-difference certificate. These results provide structural and computable criteria for assessing global optimality in quantum AB algorithms.
\end{abstract}


\section{Introduction}

Arimoto--Blahut (AB) algorithms play a central role in information theory \cite{Arimoto,Blahut}. Originating from the computation of classical channel capacity and rate-distortion functions, these algorithms provide deterministic closed-form iterations and avoid explicit computation of the objective gradient. Beyond their classical origins, AB-type methods have also been extended in quantum information theory, for example in problems involving quantum relative entropy, channel divergences, channel capacities, and quantum information bottleneck \cite{Nagaoka,Dupuis,Sutter,Li-Cai,RISB,hayashi2024ab,Hay25,HY,HSF2,SCKW}. In these problems, the appeal of such algorithms is especially clear: many objective functions involve matrix logarithms, noncommutative relative entropy, or implicit optimization over density operators, for which direct computation of the objective gradient can be prohibitively expensive.

Despite these computational advantages, the global optimality of the states returned by generalized quantum AB algorithms requires careful analysis. Monotonic decrease of the objective and numerical stabilization do not by themselves establish optimality. Even for a convex objective, an AB fixed point need not satisfy the first-order optimality condition, because the AB update direction may differ from the objective gradient. It is therefore important to identify conditions under which an AB fixed point is globally optimal and to develop practical methods for assessing the accuracy of computed iterates.

A common modern viewpoint is to interpret AB updates as special instances of entropic mirror descent \cite{HSF1,HSF2}. This interpretation is attractive because mirror descent is well studied, and therefore one may hope to import its convergence guarantees directly to AB algorithms.
However, the equivalence results in \cite{HSF1,Hay25} require specific structural conditions on $\Omega(\rho)$. In many practical problems, two features make this relationship more delicate. First, the feasible set may be subject to additional linear constraints, in which case the update involves an $e$-projection and equivalence must be analyzed relative to the feasible geometry. Second, the function $\Omega(\rho)$ used in the generalized AB algorithm \cite{hayashi2024ab} is generally introduced through the representation $f(\rho):=\tr\rho\Omega(\rho)$ rather than defined as the gradient $\nabla f(\rho)$. Although $\Omega(\rho)$ can agree with $\nabla f(\rho)$ up to normalization terms under certain assumptions, those assumptions are restrictive and need not hold in general. Without these assumptions, the mirror-descent interpretation alone does not settle the optimality of an AB fixed point. This leads to the following question:

\begin{itemize}
    \item \textit{What conditions ensure the global optimality of AB fixed points, even when the AB and MD algorithms are not equivalent?}
\end{itemize}

We answer this question by comparing the fixed-point condition of AB with the first-order optimality condition of the objective. Under linear constraints, the AB direction at a full-rank fixed point has no component along feasible tangent directions. For a convex differentiable objective, global optimality at such a point is equivalent to the objective gradient having the same property. We therefore obtain the following characterization: a full-rank AB fixed point $\rho^*$ is a global minimizer over the feasible set $\cM$ if and only if $\Omega(\rho^*)-\nabla f(\rho^*)\in N_{\cM}$, where $N_{\cM}$ is the normal space generated by normalization and the linear constraints.

This characterization also explains the role of mirror descent in assessing AB optimality. We show that an AB update at a full-rank state $\rho$ coincides with the corresponding MD update if and only if $\Omega(\rho)-\nabla f(\rho)\in N_{\cM}$. Thus, pathwise equivalence requires this compatibility condition along the entire common trajectory, while global optimality of a full-rank AB fixed point requires it only at that point. The two requirements differ in where the condition must hold. Failure of pathwise equivalence therefore does not by itself rule out global optimality of an AB fixed point.

Recently, deterministic zeroth-order mirror descent driven by general vector fields has connected finite-difference oracles, generalized Arimoto--Blahut updates, and trajectory-wise a posteriori
certification \cite{hayashi2026deterministic}. Whereas \cite{liu2026posteriori} focuses on the a posteriori certification of iterates along a given AB trajectory, the present work focuses on the structural conditions characterizing global optimality of AB fixed points and their relationship to AB--MD equivalence.
Building on the certification viewpoint developed in \cite{liu2026posteriori,hayashi2026deterministic}, we reformulate the fixed-point condition, which involves the objective gradient $\nabla f$, in terms of feasible directional derivatives and develop practical optimality certificates. At a full-rank AB fixed point, global optimality is equivalent to the vanishing of all feasible first-order variations. More generally, the distance from $\nabla f(\rho)$ to $N_{\cM}$ provides a rigorous upper bound on the objective gap, and finite differences along feasible tangent directions yield a computable conservative bound using only objective evaluations. The certification procedure therefore retains the computational appeal of the AB iteration: it does not require explicitly computing $\nabla f(\rho)$.

To illustrate the global-optimality criterion and its practical verification, we study optimization problems involving the negative quantum relative entropy between channels. For dephasing and depolarizing channels, we analytically identify the global minimizer as the unique full-rank AB fixed point in both unconstrained and linearly constrained settings, although the AB and MD updates do not agree in general. In the reported numerical experiments, the two algorithms follow different trajectories and exhibit different empirical convergence behavior over the plotted iterations, while the proposed certificates correctly bound the objective gap. Conversely, an amplitude-damping example shows numerically that a monotonically decreasing AB iteration can stabilize at a suboptimal fixed point, which the finite-difference certificate detects.
Together, these examples show that failure of pathwise equivalence is compatible with either global success or global failure; the global optimality condition at fixed points distinguishes between the two cases.

The remainder of this paper is organized as follows. Section~\ref{sec: preliminary} gives brief preliminaries and introduces the mirror-descent algorithm and the generalized quantum Arimoto--Blahut algorithm. Section~\ref{sec: equivalence} characterizes update-level and pathwise equivalence, relates accumulation points to AB fixed points, establishes the global-optimality criterion for full-rank AB fixed points, and develops a posteriori optimality certificates. Section~\ref{sec: examples} applies these results to unconstrained and linearly constrained channel-relative-entropy problems and evaluates the certificates numerically. Section~\ref{sec: counterexample} shows that monotonic convergence of the AB iteration does not guarantee global optimality and uses the certificate to identify a suboptimal fixed point. Section~\ref{sec: discussion} concludes with a discussion and open directions.

\section{Preliminaries}\label{sec: preliminary}
\subsection{Optimization Problem and Mixture Family}
In this paper, we consider the minimization problem over a finite-dimensional Hilbert space $\cH$. Let $\cB(\cH)$ be the set of Hermitian matrices on $\cH$, and let $\cS(\cH)$ be the set of density matrices on $\cH$. Given a continuous function $\Omega: \cS(\cH)\to \cB(\cH)$, we consider the following minimization problem
\begin{align}
\label{problem: unconstrained}
    \min_{\rho\in\cS(\cH)} \tr\rho\Omega(\rho).
\end{align}
The objective function $f$ is defined as $f(\rho) = \tr \rho\Omega(\rho)$. We note that many objective functions in quantum information theory can be written in the form \eqref{problem: unconstrained}. For example, for a relative-entropy-type objective function, the function $\Omega$ can be composed of logarithmic, linear, and constant terms in $\rho$. Moreover, this problem formulation covers the optimization of a quantum operation, since one can transform an optimization problem over quantum operations into an optimization problem over quantum states on an enlarged Hilbert space by using purification \cite{hayashi2024ab}.

While the above problem considers optimization over $\cS(\cH)$, practical problems usually contain additional linear constraints. To handle these linear constraints, we use the concept of a mixture family \cite{Hmixture}. Given $k$ linearly independent Hermitian matrices $H_1,\ldots,H_k$ on $\cH$ and a constant vector $c = (c_1,\ldots,c_k)^T\in\mathbb{R}^k$, we define the mixture family $\cM$ as follows:
\begin{align}
    \cM :=\{\rho\in \cS(\cH)| \tr\rho H_j = c_j \quad \text{for } j = 1,\ldots,k\}.
\end{align}
With linear constraints, the minimization problem is then formulated as follows:
\begin{align}
    \min_{\rho\in\cM}\tr\rho\Omega(\rho).
\end{align}

\subsection{Mirror Descent Algorithm}
Although gradient descent is a powerful algorithm that assumes a Euclidean parameter space, mirror descent \cite{beck2003mirror} is more suitable for quantum information theory because it replaces the Euclidean norm with a Bregman divergence, which is more appropriate for the natural geometry of quantum states. Let \(f:\mathcal S(\mathcal H)\to \mathbb R\) be a differentiable objective function. Given the current
iterate \(\rho^{(t)}\in\cM\), the mirror descent update with
stepsize \(1/\gamma\) is defined by
\begin{align}
    \rho^{(t+1)}
=
\argmin_{\sigma\in\cM}
\left\{
\left\langle \nabla f(\rho^{(t)}),\sigma-\rho^{(t)}\right\rangle
+
\gamma D_{\varphi}(\sigma\|\rho^{(t)})
\right\},
\end{align}
where $D_{\varphi}(\cdot\|\cdot)$ is the Bregman divergence, and the most natural choice of Bregman divergence in quantum information theory is quantum relative entropy. Thus, the MD update is
\begin{align}
    \rho^{(t+1)}
=
\argmin_{\sigma\in\cM}
\left\{
\left\langle \nabla f(\rho^{(t)}),\sigma-\rho^{(t)}\right\rangle
+
\gamma D(\sigma\|\rho^{(t)})
\right\},
\end{align}
where $D(\sigma\|\rho) = \tr\sigma(\log\sigma - \log \rho)$.

Since
\(\langle \nabla f(\rho^{(t)}),\rho^{(t)}\rangle\) is independent of
\(\sigma\), the update can also be written as
\begin{align}
    \rho^{(t+1)}
=
\argmin_{\sigma\in\cM}
\left\{
\tr\sigma\nabla f(\rho^{(t)})
+
\gamma D(\sigma\|\rho^{(t)})
\right\}.
\end{align}
The mirror descent algorithm is shown in Algorithm~\ref{AL-MD}.

\begin{algorithm}
\caption{Mirror Descent Algorithm}
\Label{AL-MD}
\begin{algorithmic}
\STATE {Choose the initial value $\rho^{(0)} \in \mathcal{M}$;} 
\REPEAT 
\STATE Calculate
\STATE $\displaystyle \rho^{(t+1)}
=
\argmin_{\sigma\in\cM}
\left\{
\tr\sigma\nabla f(\rho^{(t)})
+
\gamma D(\sigma\|\rho^{(t)})
\right\}$;
\UNTIL{convergence.} 
\end{algorithmic}
\end{algorithm}

Mirror descent performs a proximal step with respect to the
quantum relative entropy and uses the objective gradient \(\nabla f(\rho^{(t)})\) as the update direction \cite{tsuda2005matrix}. The generalized Arimoto--Blahut algorithm has a similar updating rule, but its update direction is given by \(\Omega(\rho^{(t)})\) rather than by \(\nabla f(\rho^{(t)})\).

\subsection{Generalized Quantum Arimoto--Blahut Algorithm}
The Arimoto--Blahut algorithm is a well-known alternating algorithm with broad applications in quantum information theory. To describe the AB algorithm, we introduce the $e$-projection of $\rho$ onto ${\cal M}$, which is defined as
$\Gamma^{(e)}_{{\cal M}}[\rho]:=
\argmin_{ \sigma \in {\cal M}}
D(\sigma\|\rho)$ \cite{Amari-Nagaoka}.
Calculation of the $e$-projection $\Gamma^{(e)}_{{\cal M}}[\rho]$ relies on the solution of the following equations:
\begin{align}
\frac{\partial}{\partial \tau^i}
\log \Tr \exp
(\log \rho + \sum_{j=1}^k H_j \tau^j )
=c_i.
\end{align}
If the solution of the above equations is $\tau_*=(\tau_*^1, \ldots, \tau_*^k)$, then the $e$-projection $\Gamma^{(e)}_{{\cal M}}[\rho]$ is given by
$    \Gamma^{(e)}_{{\cal M}}[\rho] = C \exp (\log \rho + \sum_{j=1}^k H_j \tau_*^j )$,
where $C$ is a normalizing constant~\cite{H23}.
Instead of solving these equations directly, their solution can be obtained by considering the following minimization problem:
\begin{align}
\tau_*:=
\argmin_{\tau }
\log \Big(\Tr \Big(\exp
(\log \rho + \sum_{j=1}^k H_j \tau^j )\Big)\Big)
-\sum_{i=1}^k \tau^i c_i.\Label{BNB}
\end{align}
We note that the above function
is convex for $\tau$ \cite[Section III-C]{H23}.

The quantum Arimoto--Blahut algorithm is shown in Algorithm~\ref{AL1}. We employ
the conversion function $\cF_3[\cdot]:\cS(\cH)\to\cS(\cH)$, defined by
${\cal F}_3[\sigma]:= \frac{1}{\kappa[\sigma]}
\exp( \log \sigma -\frac{1}{\gamma} \Omega[\sigma])$,
where $\kappa[\sigma]$ is the normalization factor of
$\Tr \exp( \log \sigma -\frac{1}{\gamma} \Omega[\sigma])$.

\begin{algorithm}
\caption{Quantum AB Algorithm}
\Label{AL1}
\begin{algorithmic}
\STATE {Choose the initial value $\rho^{(0)} \in \mathcal{M}$;} 
\REPEAT 
\STATE Calculate
\STATE $\displaystyle \rho^{(t+1)}:=\Gamma^{(e)}_{{\cal M}}[{\cal F}_3[\rho^{(t)}]]$;
\UNTIL{convergence.} 
\end{algorithmic}
\end{algorithm}

\section{Global Optimality Conditions and Mirror-Descent Equivalence}\label{sec: equivalence}
This section develops conditions for the global optimality of full-rank AB fixed points and methods for certifying the accuracy of AB iterates. We begin with the geometry induced by normalization and linear constraints and use it to characterize agreement between AB and MD updates. We then relate full-rank accumulation points to AB fixed points, establish the global-optimality criterion for convex differentiable objectives, and derive a posteriori certificates. The analysis identifies how the condition for fixed-point optimality differs from the condition for pathwise equivalence to MD.

\subsection{Geometric Setting and Update Maps}

Throughout this section, let
\begin{align}
\cM:=\{\rho\in\cS(\cH):\tr(\rho H_j)=c_j,\ j=1,\ldots,k\},
\end{align}
where $\{I,H_1,\ldots,H_k\}$ is linearly independent and $\cM$ contains a positive-definite state.  We work on the relative interior
\begin{align}
\cM_{++}:=\{\rho\in\cM:\rho\succ0\},
\end{align}
which is the natural domain of the logarithmic updates. All results in this section concern full-rank iterates in
$\mathcal M_{++}$. With respect to the Hilbert--Schmidt inner product, the tangent and normal spaces of the affine hull of $\cM$ are
\begin{align}
T_{\cM}&:=\{X\in\cB(\cH):\tr X=0,\notag\\
&\qquad \tr(XH_j)=0,\ j=1,\ldots,k\},\label{eq:tangent-space}\\
N_{\cM}&:=T_{\cM}^{\perp} =\operatorname{span}\{I,H_1,\ldots,H_k\}.\label{eq:normal-space}
\end{align}
Assume that $f(\rho)=\tr[\rho\Omega(\rho)]$ is differentiable relative to the affine hull of $\cM$.  For simplicity, we denote the update maps of the AB algorithm and the MD algorithm by $\cT_{AB}(\cdot)$ and $\cT_{MD}(\cdot)$, respectively. Thus for $\gamma>0$, we have
\begin{align}
\cT_{AB}(\rho)&:=\Gamma_{\cM}^{(e)}[\cF_3[\rho]],\label{eq:AB-map}\\
\cT_{MD}(\rho)&:=\argmin_{\sigma\in\cM}
\{\tr[\sigma\nabla f(\rho)]+\gamma D(\sigma\|\rho)\}.\label{eq:MD-map}
\end{align}
Let $P_T$ be the Hilbert--Schmidt orthogonal projection onto $T_{\cM}$.


\subsection{Update-Level and Pathwise Equivalence}
In previous work, most analyses of the convergence and global optimality of Algorithm~\ref{AL1} rely on the following two conditions:
\begin{itemize}
    \item[(a1)] All pairs $(\rho^{(t+1)},\rho^{(t)})$ satisfy the following condition with $(\rho,\sigma) = (\rho^{(t+1)},\rho^{(t)})$
    \begin{align}
        D_{\Omega}(\rho\|\sigma):= \tr\rho[\Omega(\rho) - \Omega(\sigma)] \leq \gamma D(\rho\|\sigma), \label{eq: condition a1}
    \end{align}
    for a sufficiently large $\gamma$. \label{condition: convergence 1}

    \item[(a2)] For any $\rho,\sigma\in\cM$, we have $D_{\Omega}(\rho\|\sigma)\geq 0$.\label{condition: convergence 2} 
\end{itemize}

Under condition (a1), \cite[Theorem 3.3]{RISB} and \cite[Theorem 1]{hayashi2024ab} guarantee that the algorithm will iteratively decrease the value of the objective function. Thus, this property guarantees convergence of the objective values of Algorithm~\ref{AL1}. However, Algorithm~\ref{AL1} may become trapped in a local minimum if no additional conditions are imposed. To ensure the global optimality of Algorithm~\ref{AL1}, it has been shown in \cite{hayashi2024ab,HSF1,RISB} that the objective value converges to the global optimum if condition (a2) also holds. Under these conditions, \cite{HSF1} proposed the following proposition:

\begin{proposition}\cite{HSF1}\label{proposition: equivalence by Kerry}
    Consider a continuous function $\Omega: \cS(\cH)\to \cB(\cH)$ satisfying conditions (a1) and (a2). Define the function
    \begin{align}
        f: \cB(\cH)\to \mathbb{R},\quad f(\rho) = \tr\rho\Omega(\rho). 
    \end{align}
    Then $f$ is a convex and differentiable function on $\cB(\cH)$, and the AB iterates are identical to the MD iterates with step size $t_k = \frac{1}{\gamma}$.
\end{proposition}

Proposition~\ref{proposition: equivalence by Kerry} tells us that, when conditions (a1) and (a2) hold for $\Omega$, the AB algorithm has the same iterates as MD, making the two algorithms equivalent. However, although condition (a1) can be satisfied by choosing a sufficiently large value of $\gamma$ and can be easily verified by evaluating the inequality \eqref{eq: condition a1} along the computed trajectory, condition (a2) is almost impossible to verify numerically, and an analytical proof is frequently intractable in practical applications. In fact, condition (a2) does not necessarily hold in general, leading to the failure of the above equivalence result and global optimality analysis.

To understand the relationship between the AB algorithm and the MD algorithm, we reconsider this problem from a different perspective. Rather than deducing equivalence from global assumptions on $\Omega$, we compare the two proximal subproblems at a fixed state $\rho$. The next lemma shows that the $e$-projection formulation of the AB map admits exactly the same variational form, with $\Omega(\rho)$ replacing the objective gradient $\nabla f(\rho)$.

\begin{lemma}\label{lemma:variational-AB}
For every $\rho\in\cM_{++}$,
\begin{align}
\cT_{AB}(\rho)=\argmin_{\sigma\in\cM}
\{\tr[\sigma\Omega(\rho)]+\gamma D(\sigma\|\rho)\}.
\label{eq:AB-proximal-form}
\end{align}
\end{lemma}

\begin{proof}
By definition,
$\log\cF_3[\rho]=\log\rho-\gamma^{-1}\Omega(\rho)-\log\kappa[\rho]$.
Consequently, for every $\sigma\in\cM$,
\begin{align*}
    D(\sigma\|\cF_3[\rho]) &= \tr\sigma(\log \sigma - \log \cF_3[\rho])\\
    &= \tr\sigma(\log\sigma - \log\rho +\frac{1}{\gamma}\Omega(\rho)+\log\kappa[\rho])\\
    &= D(\sigma\|\rho) + \frac{1}{\gamma}\tr\sigma\Omega(\rho) +\log\kappa[\rho].
\end{align*}
The last term is independent of $\sigma$.  Multiplication by $\gamma$ therefore leaves the minimizer unchanged and gives \eqref{eq:AB-proximal-form}.
\end{proof}

\begin{lemma}[Entropic proximal minimizers]
\label{lemma:proximal-well-posedness}
Let $\rho\in\cM_{++}$, $A\in\cB(\cH)$, and $\gamma>0$. Then the problem
\begin{align}
\min_{\sigma\in\cM}
\left\{\tr(\sigma A)+\gamma D(\sigma\|\rho)\right\}
\label{eq:general-proximal-problem}
\end{align}
has a unique minimizer, and this minimizer belongs to $\cM_{++}$.
\end{lemma}

\begin{proof}
Since $\cM$ is compact and the objective in \eqref{eq:general-proximal-problem} is continuous in $\sigma$, a minimizer exists. The objective is strictly convex in $\sigma$ because $\sigma\mapsto\tr(\sigma\log\sigma)$ is strictly convex. Hence the minimizer is unique.

Let $\widehat\sigma$ denote the minimizer and suppose, for contradiction, that it is singular. Since $\rho\in\cM_{++}$, the state
\begin{align}
\sigma_t:=(1-t)\widehat\sigma+t\rho,
\qquad 0<t<1,
\end{align}
belongs to $\cM_{++}$. Because $\rho$ is strictly positive on the kernel of $\widehat\sigma$, the right derivative of $\tr(\sigma_t\log\sigma_t)$ at $t=0$ is $-\infty$, whereas the remaining terms in the objective have finite right derivatives. The objective therefore decreases for all sufficiently small $t>0$, contradicting the optimality of $\widehat\sigma$. Thus $\widehat\sigma\in\cM_{++}$.
\end{proof}

Applying Lemma~\ref{lemma:proximal-well-posedness} with $A=\Omega(\rho)$ and $A=\nabla f(\rho)$ shows that both the AB and MD proximal subproblems have unique positive-definite minimizers.

Then we have the following theorem:
\begin{theorem}[One-step equivalence]\label{theorem: equivalence of each step}
    For every $\rho\in \cM_{++}$, the following statements are equivalent:
    \begin{itemize}
        \item[(b1)] $\cT_{AB}(\rho) = \cT_{MD}(\rho)$.
        \item[(b2)] $\Omega(\rho) - \nabla f(\rho) \in N_{\cM}$.
    \end{itemize}
\end{theorem}

\begin{proof}
    Suppose first that (b2) holds. Then
    \begin{align}
    \Omega(\rho)-\nabla f(\rho)=a_0I+\sum_{j=1}^k a_jH_j,
    \end{align}
    for some real coefficients $a_0,\ldots,a_k$.
    
    Then for every $\sigma \in \cM$,
    \begin{align}
        \tr[\sigma (\Omega(\rho) - \nabla f(\rho))] =& a_0 \tr\sigma + \sum_{j=1}^k a_j\tr(\sigma H_j) \notag\\
         = &a_0 + \sum_{j=1}^k a_j c_j,
    \end{align}
    which is a constant on $\cM$.
    Hence, the two objectives
    \begin{align}
        \sigma &\mapsto \tr[\sigma\Omega(\rho)] + \gamma D(\sigma\|\rho) \label{objective: AB}\\
        \sigma &\mapsto \tr[\sigma\nabla f(\rho)] + \gamma D(\sigma\|\rho)\label{objective: MD}
    \end{align}
    differ only by a constant on the feasible set $\cM$. Therefore they have the same minimizer, i.e., $\cT_{AB}(\rho) = \cT_{MD}(\rho)$. Thus the condition (b2) implies condition (b1).

    Conversely, suppose that $\cT_{AB}(\rho) = \cT_{MD}(\rho) = \widehat\rho$. By Lemma~\ref{lemma:proximal-well-posedness}, the common minimizer $\widehat\rho$ belongs to $\cM_{++}$. Therefore, the logarithmic first-order optimality conditions used below are well defined.
    
    The AB step solves
    \begin{align}
        \min_{\sigma\in\cM} \Bigl\{\tr\sigma\Omega(\rho) + \gamma D(\sigma\|\rho)\Bigr\}.
    \end{align}
    The Lagrangian is
    \begin{align*}
        \cL_{AB}(\sigma,\alpha,\tau)
        ={}&\tr\sigma\Omega(\rho)+\gamma D(\sigma\|\rho)\\
        &+\alpha(\tr\sigma-1)
        +\sum_{i=1}^k\tau_i(\tr\sigma H_i-c_i).
    \end{align*}
    Differentiating with respect to $\sigma$ and setting the derivative to zero at $\widehat\rho$ gives
    \begin{align}
        \Omega(\rho) + \gamma (\log\widehat\rho - \log\rho+I) + \alpha I +\sum_{i=1}^k\tau_i H_i = 0.
    \end{align}
    Equivalently,
    \begin{align}
        \log\widehat\rho = \log\rho -\frac{1}{\gamma}\Omega(\rho) + u_0I + \sum_{i=1}^k u_i H_i,
    \end{align}
    for some real coefficients $u_0,u_1,\ldots,u_k$.
    
    Similarly, the MD step solves
    \begin{align}
         \min_{\sigma\in\cM} \Bigl\{\tr\sigma\nabla f(\rho) + \gamma D(\sigma\|\rho)\Bigr\}.
    \end{align}
    Thus, at the minimizer $\widehat\rho$, we have
    \begin{align}
        \log\widehat\rho = \log\rho -\frac{1}{\gamma}\nabla f(\rho) + v_0 I + \sum_{i=1}^k v_i H_i,
    \end{align}
    for some real coefficients $v_0,v_1,\ldots,v_k$.

    Therefore
    \begin{align*}
       &\log\rho -\frac{1}{\gamma}\Omega(\rho) + u_0I + \sum_{i=1}^k u_i H_i \\
       =& \log\rho -\frac{1}{\gamma}\nabla f(\rho) + v_0 I + \sum_{i=1}^k v_i H_i.
    \end{align*}
    Then
    \begin{align}
        \frac{1}{\gamma}(\Omega(\rho) - \nabla f(\rho)) = (u_0-v_0)I + \sum_{i=1}^k (u_i-v_i) H_i,
    \end{align}
    hence
    \begin{align}
        \Omega(\rho) - \nabla f(\rho) \in \operatorname{span}\{I,H_1,\ldots,H_k\} = N_{\cM}.
    \end{align}
    Thus, the condition (b1) implies condition (b2).

\end{proof}

\begin{remark}
    Condition (b2) is equivalent to equality of the tangent components of the AB and MD update directions, i.e.,
    \begin{align}
        \Omega(\rho) - \nabla f(\rho)\in N_{\cM} \iff P_T\Omega(\rho) = P_T\nabla f(\rho).
    \end{align}
    Thus, the AB and MD updates agree precisely when their update directions have the same feasible tangent component.
\end{remark}
Applying the one-step criterion recursively along a common trajectory immediately yields the pathwise statement.

\begin{corollary}[Pathwise equivalence]\label{cor:pathwise-equivalence}
Initialize AB and MD at the same $\rho^{(0)}\in\cM_{++}$.  Their iterates agree through step $T$ if and only if
\begin{align}
\Omega(\rho^{(t)})-\nabla f(\rho^{(t)})\in N_{\cM},
\qquad t=0,\ldots,T-1,
\end{align}
along the common iterates.  In particular, $\cT_{AB}=\cT_{MD}$ on $\cM_{++}$ if and only if the condition (b2) holds for every $\rho\in\cM_{++}$.
\end{corollary}

Corollary~\ref{cor:pathwise-equivalence} tells us that, if $\Omega(\rho) - \nabla f(\rho) \in N_{\cM}$ holds for every $\rho\in\cM_{++}$, then the AB algorithm and the MD algorithm generate exactly the same sequence from the same initialization. That is, the AB algorithm is pathwise equivalent to the MD algorithm. We note that Proposition~\ref{proposition: equivalence by Kerry} gives a sufficient but not necessary condition for pathwise equivalence. Under conditions (a1) and (a2), one can prove that $\nabla f(\rho) = \Omega(\rho)$ for all $\rho\in \cM$, which directly leads to condition (b2). Theorem~\ref{theorem: equivalence of each step} isolates the exact algebraic requirement for equivalence. Only the component of the update direction tangent to the feasible affine space matters. Thus, equality of $\Omega(\rho)$ and $\nabla f(\rho)$ is sufficient but unnecessarily strong. A failure of this condition at any iterate separates the two trajectories from that step onward. It does not determine whether the AB trajectory can later reach a globally optimal state.

\subsection{Accumulation Points and Fixed Points}

The preceding analysis characterizes when an individual AB update coincides with a mirror-descent update. To connect these update-level results with the asymptotic behavior of the AB iteration, one must distinguish convergence of the objective values from convergence of the state sequence. In particular, monotonic convergence of $f(\rho^{(t)})$ does not by itself imply that the sequence $\{\rho^{(t)}\}$ converges to a single state.

We therefore formulate the asymptotic behavior of the AB iteration in terms of its accumulation points. For an initial state $\rho^{(0)}$, define the accumulation set of the AB trajectory by
\begin{align}
\omega(\rho^{(0)})
:=\bigl\{\bar{\rho}\in\mathcal{M}:&\
\rho^{(t_k)}\rightarrow\bar{\rho}\notag\\[-1mm]
&\text{for some subsequence }t_k\rightarrow\infty\bigr\}.
\end{align}
Since $\mathcal{M}$ is compact, $\omega(\rho^{(0)})$ is nonempty. In this work, we restrict the fixed-point analysis to accumulation points belonging to the relative interior $\mathcal{M}_{++}$, which is the natural domain of the logarithmic AB update. The following result shows that condition (a1) also implies asymptotic regularity of the AB iteration.

\begin{proposition}
\label{prop:asymptotic_regular}
Let $\{\rho^{(t)}\}\subset\mathcal{M}_{++}$ be generated by the AB iteration and suppose that condition (a1) holds along the trajectory, i.e.,
\begin{align}
D_{\Omega}(\rho^{(t+1)}\Vert\rho^{(t)})\leq\gamma D(\rho^{(t+1)}\Vert\rho^{(t)}).
\end{align}
Then
\begin{align}
D(\rho^{(t)}\Vert\rho^{(t+1)})\rightarrow0.
\end{align}
In particular,
\begin{align}
\|\rho^{(t)}-\rho^{(t+1)}\|_1\rightarrow0.\label{eq:asymptotic_regular}
\end{align}
\end{proposition}

\begin{proof}
From the proof of \cite[Theorem 1]{hayashi2024ab}, we know that condition (a1) implies 
    \begin{align}
        D(\rho^{(t)}\|\rho^{(t+1)}) \leq \frac{f(\rho^{(t)}) - f(\rho^{(t+1)})}{\gamma}.
    \end{align}
    Since $f$ is bounded below on $\mathcal{M}$, summing the preceding inequality gives
    \begin{align*}
        \sum_{t=0}^{\infty} D(\rho^{(t)}\|\rho^{(t+1)}) &\leq \sum_{t=0}^{\infty} \frac{f(\rho^{(t)}) - f(\rho^{(t+1)})}{\gamma} \\
        &\leq \frac{f(\rho^{(0)})-\inf_{\rho\in\cM}f(\rho)}{\gamma}\\
        &<\infty.
    \end{align*}
    Hence,
    \begin{align}
        D(\rho^{(t)}\|\rho^{(t+1)}) \rightarrow 0.
    \end{align}
    By the quantum Pinsker inequality, we have
    \begin{align}
        D(\rho^{(t)}\|\rho^{(t+1)}) \geq \frac{1}{2}\|\rho^{(t)} - \rho^{(t+1)}\|_1^2,
    \end{align}
    so
    \begin{align}
        \|\rho^{(t)} - \rho^{(t+1)}\|_1 \rightarrow 0.
    \end{align}
\end{proof}

The preceding proposition shows that condition (a1) implies asymptotic regularity of the state sequence. Asymptotic regularity does not by itself imply convergence to a unique state, but it is sufficient to connect accumulation points of the trajectory with fixed points of the AB map.

We first note that, under the standing continuity assumption on $\Omega$, the AB update map is continuous on $\mathcal{M}_{++}$.

\begin{lemma}[Continuity of the AB update map]
\label{lemma:ab_continuity}
The map
\begin{align}
\mathcal{T}_{\mathrm{AB}}: \mathcal{M}_{++}\rightarrow\mathcal{M}_{++}
\end{align}
is continuous.
\end{lemma}

\begin{proof}
By Lemma~\ref{lemma:proximal-well-posedness}, the AB minimizer is unique and belongs to $\mathcal{M}_{++}$ for every $\rho\in\mathcal{M}_{++}$.
Let $\rho_n\rightarrow\rho\in\mathcal{M}_{++}$ and define
\begin{align}
\sigma_n:=\mathcal{T}_{\mathrm{AB}}(\rho_n).
\end{align}
Since $\mathcal{M}$ is compact, every subsequence of $\{\sigma_n\}$ contains a further convergent subsequence. Let
\begin{align}
\sigma_{n_k}\rightarrow\bar{\sigma}.
\end{align}
For every $\sigma\in\mathcal{M}$, optimality of $\sigma_{n_k}$ gives
\begin{align}
&\Tr[\sigma_{n_k}\Omega(\rho_{n_k})]
 +\gamma D(\sigma_{n_k}\Vert\rho_{n_k})\notag\\
&\qquad\leq \Tr[\sigma\Omega(\rho_{n_k})]
 +\gamma D(\sigma\Vert\rho_{n_k}).
\end{align}
Since $\rho_{n_k}\rightarrow\rho\succ0$, the states $\rho_{n_k}$ are uniformly positive definite for all sufficiently large $k$. In finite dimensions, $D(\sigma\Vert\rho)$ is then continuous in both arguments even when $\sigma$ is singular: $\Tr(\sigma\log\sigma)$ is continuous up to the boundary under the convention $0\log0=0$, and $\log\rho_{n_k}\rightarrow\log\rho$. Taking $k\rightarrow\infty$ therefore gives
\begin{align}
\Tr[\bar{\sigma}\Omega(\rho)] + \gamma D(\bar{\sigma}\Vert\rho) \leq \Tr[\sigma\Omega(\rho)] + \gamma D(\sigma\Vert\rho),
\end{align}
for every $\sigma\in\mathcal{M}$. Hence
\begin{align}
\bar{\sigma}=\mathcal{T}_{\mathrm{AB}}(\rho),
\end{align}
where uniqueness follows from Lemma~\ref{lemma:proximal-well-posedness}. Therefore every convergent subsequence of $\{\sigma_n\}$ has the same limit, which implies
\begin{align}
\mathcal{T}_{\mathrm{AB}}(\rho_n) \rightarrow \mathcal{T}_{\mathrm{AB}}(\rho).
\end{align}
\end{proof}

We can now relate the actual asymptotic behavior of the AB trajectory to its fixed points.

\begin{theorem}[Accumulation points are AB fixed points]
\label{thm:accumulation_fixed}
Suppose condition \emph{(a1)} holds along the AB trajectory. If
\begin{align}
\rho^{(t_k)}\rightarrow\bar{\rho},
\end{align}
for some accumulation point $\bar{\rho}\in\mathcal{M}_{++}$, then
\begin{align}
\mathcal{T}_{\mathrm{AB}}(\bar{\rho}) = \bar{\rho}.
\end{align}
\end{theorem}

\begin{proof}
By Proposition~\ref{prop:asymptotic_regular},
\begin{align}
\|\rho^{(t+1)}-\rho^{(t)}\|_1\rightarrow0.
\end{align}
Therefore,
\begin{align}
\|\rho^{(t_k+1)}-\bar{\rho}\|_1 &\leq \|\rho^{(t_k+1)}-\rho^{(t_k)}\|_1 + \|\rho^{(t_k)}-\bar{\rho}\|_1 \nonumber\\ &\rightarrow0.
\end{align}
Hence
\begin{align}
\rho^{(t_k+1)}\rightarrow\bar{\rho}.
\end{align}
On the other hand,
\begin{align}
\rho^{(t_k+1)} = \mathcal{T}_{\mathrm{AB}}(\rho^{(t_k)}).
\end{align}
By Lemma~\ref{lemma:ab_continuity},
\begin{align}
\mathcal{T}_{\mathrm{AB}}(\bar{\rho}) &= \lim_{k\rightarrow\infty} \mathcal{T}_{\mathrm{AB}}(\rho^{(t_k)}) \nonumber\\ &= \lim_{k\rightarrow\infty} \rho^{(t_k+1)} = \bar{\rho}.
\end{align}
Thus, $\bar{\rho}$ is an AB fixed point.
\end{proof}

Theorem~\ref{thm:accumulation_fixed} provides the link between the conventional trajectory analysis and the fixed-point analysis developed below. Under the assumption of condition (a1), every relative-interior accumulation point must satisfy the AB fixed-point condition. If uniqueness of the AB fixed point can additionally be established, state convergence follows.

\begin{corollary}
\label{cor:unique_fixed_state_convergence}
Suppose condition (a1) holds along the AB trajectory and assume that every accumulation point of $\{\rho^{(t)}\}$ belongs to $\mathcal{M}_{++}$. If the AB map has a unique fixed point $\rho^*\in\mathcal{M}_{++}$, then
\begin{align}
\rho^{(t)}\rightarrow\rho^*.
\end{align}
\end{corollary}

\begin{proof}
Since $\mathcal{M}$ is compact, the sequence has at least one accumulation point. By Theorem~\ref{thm:accumulation_fixed}, every accumulation point is an AB fixed point. Uniqueness therefore implies
\begin{align}
\omega(\rho^{(0)})=\{\rho^*\}.
\end{align}
Suppose, for contradiction, that $\rho^{(t)}$ does not converge to $\rho^*$. Then there exist $\varepsilon>0$ and a subsequence $\rho^{(t_k)}$ such that
\begin{align}
\|\rho^{(t_k)}-\rho^*\|_1\geq\varepsilon.
\end{align}
Compactness provides a further convergent subsequence
\begin{align}
\rho^{(t_{k_\ell})}\rightarrow\widetilde{\rho}.
\end{align}
The state $\widetilde{\rho}$ is an accumulation point, and therefore $\widetilde{\rho}=\rho^*$, which contradicts the preceding inequality. Hence
\begin{align}
\rho^{(t)}\rightarrow\rho^*.
\end{align}
\end{proof}

\begin{remark}
The distinction between convergence to an accumulation set and convergence to a particular state is important for interpreting finite AB iterations. For any sequence in the compact feasible set,
\begin{align}
\operatorname{dist} \bigl( \rho^{(t)},\omega(\rho^{(0)}) \bigr) \rightarrow0.
\end{align}
Thus, a sufficiently late iterate approaches the accumulation set. In numerical implementations, a finite iterate can be regarded as an approximate AB fixed point when $\|\rho^{(t+1)}-\rho^{(t)}\|_1$ is small.
\end{remark}

\subsection{Global Optimality of Fixed Points}

We now characterize when a full-rank AB fixed point is globally optimal. For convex differentiable objectives, we compare the AB fixed-point condition with the first-order optimality condition and explain how their relationship differs from pathwise equivalence to MD. We first characterize the fixed points of the two update maps.

\begin{lemma}[AB fixed points]\label{lemma: fixed point of AB}
For $\rho\in\cM_{++}$,
\begin{align}
\cT_{AB}(\rho)=\rho
\quad\Longleftrightarrow\quad
\Omega(\rho)\in N_{\cM}.
\end{align}
\end{lemma}

\begin{proof}
    If $\rho$ is a fixed point, the optimality condition for \eqref{eq:AB-proximal-form} at $\sigma = \rho$ is
    \begin{align}
    \Omega(\rho)+\gamma I+a_0I+\sum_{j=1}^k a_jH_j=0,
    \end{align}
    for suitable multipliers, and hence $\Omega(\rho)\in N_{\cM}$.
    Conversely, we assume that $\Omega(\rho) = a_0 I + \sum_{j=1}^k a_jH_j$. Then,
\begin{align}
    D(\sigma\|\cF_3[\rho]) &=  D(\sigma\|\rho) + \frac{1}{\gamma}\tr\sigma\Omega(\rho) +\log\kappa[\rho] \\
    &=D(\sigma\|\rho) + \frac{1}{\gamma}\tr\sigma(a_0 I + \sum_{j=1}^k a_jH_j) +\log\kappa[\rho]\\
    &=D(\sigma\|\rho) + \frac{1}{\gamma}(a_0 + \sum_{j=1}^k a_jc_j) +\log\kappa[\rho].
\end{align}
Since $ \frac{1}{\gamma}(a_0 + \sum_{j=1}^k a_jc_j) +\log\kappa[\rho]$ is independent of $\sigma$,
\begin{align}
        \cT_{AB}(\rho)=  \argmin_{\sigma\in\cM}\Bigl\{D(\sigma\|\rho)\Bigr\} = \rho.
    \end{align}
\end{proof}

\begin{lemma}[MD fixed points and optimality]\label{lemma: fixed point of MD}
Let $\rho\in\cM_{++}$.  Then
\begin{align}
\cT_{MD}(\rho)=\rho
\quad\Longleftrightarrow\quad
\nabla f(\rho)\in N_{\cM}.
\label{eq:MD-fixed-characterization}
\end{align}
If, in addition, $f$ is convex on $\cM$, these conditions are equivalent to $\rho\in\argmin_{\sigma\in\cM}f(\sigma)$.
\end{lemma}

\begin{proof}
The fixed-point equivalence follows from the first-order optimality condition for \eqref{eq:MD-map}, exactly as in Lemma~\ref{lemma: fixed point of AB}.  Now assume convexity.  If $\nabla f(\rho)\in N_{\cM}$, then for every $\sigma\in\cM$,
\begin{align}
f(\sigma)&\geq f(\rho)+\tr[\nabla f(\rho)(\sigma-\rho)]=f(\rho),
\end{align}
    because $\sigma-\rho\in T_{\cM}$.  Thus $\rho$ is a global minimizer.  Conversely, a differentiable convex function has zero directional derivative in every feasible tangent direction at a relative-interior minimizer.  Therefore $\nabla f(\rho)\perp T_{\cM}$, which is equivalent to $\nabla f(\rho)\in N_{\cM}$.
\end{proof}

At an AB fixed point, Lemma~\ref{lemma: fixed point of AB} already guarantees that $\Omega(\rho^*)$ is normal to $\cM$. Therefore, the only remaining question is whether the objective gradient is also normal. Equivalently, one needs the AB--MD compatibility condition only at the candidate terminal state, rather than at every state visited by the algorithm. We can now state the main fixed-point result.

\begin{theorem}[Global optimality of the AB fixed point]\label{theorem: equivalence of fixed point}
    Assume the objective function $f(\rho) = \tr\rho\Omega(\rho)$ is convex and differentiable, and let $\rho^*\in \cM_{++}$ be an AB fixed point. Then the following statements are equivalent:
    \begin{itemize}
        \item[(c1)] $\rho^*$ is a global minimizer of $f$ over $\cM$.
        \item[(c2)] $\rho^*$ is a fixed point of the MD algorithm.
        \item[(c3)] $\Omega(\rho^*) - \nabla f(\rho^*)\in N_{\cM}$.
    \end{itemize}
\end{theorem}

\begin{proof}
Lemma~\ref{lemma: fixed point of MD} gives the equivalence of (c1), (c2), and
$\nabla f(\rho^*)\in N_{\cM}$.  Because $\rho^*$ is an AB fixed point, Lemma~\ref{lemma: fixed point of AB} gives $\Omega(\rho^*)\in N_{\cM}$.  Since $N_{\cM}$ is a vector space,
\begin{align}
\nabla f(\rho^*)\in N_{\cM}
\quad\Longleftrightarrow\quad
\Omega(\rho^*)-\nabla f(\rho^*)\in N_{\cM},
\end{align}
which proves the result.
\end{proof}

Theorem~\ref{theorem: equivalence of fixed point} pinpoints the logical gap between pathwise equivalence and global optimality.  The compatibility condition need only hold at the limiting AB fixed point to certify that point as globally optimal; it need not hold at earlier iterates. Therefore, while the pathwise equivalence between the AB algorithm and the MD algorithm often fails in practical problems, the failure of pathwise equivalence does not automatically imply that an AB fixed point is suboptimal. The AB algorithm can still achieve the global optimizer when the fixed point satisfies condition (c3). 

Let $P_T$ be the Hilbert--Schmidt orthogonal projection onto $T_{\cM}$. To evaluate the condition (c3), we consider the following distance between an operator $Y$ and the normal space $N_{\cM}$.
\begin{align}
\dist_2(Y,N_{\cM})=\|P_TY\|_2
=\sup_{\substack{X\in T_{\cM}\\\|X\|_2\leq1}}|\tr(XY)|.\label{eq:dual-distance}
\end{align}

Thus, applying \eqref{eq:dual-distance} to $Y=\Omega(\rho)-\nabla f(\rho)$ gives:
\begin{align}
    &\operatorname{dist}_2\!\left(\Omega(\rho)-\nabla f(\rho),N_{\cM}\right) \notag\\
    &=\sup_{\substack{X\in T_{\cM}\\ \lVert X\rVert_2\leq 1}}
    \left|Df(\rho)[X]-\tr(X\Omega(\rho))\right| \notag\\
    &=\sup_{\substack{X\in T_{\cM}\\ \lVert X\rVert_2\leq 1}}
    \left|\tr(\rho D\Omega(\rho)[X])\right|,
    \label{eq:weak-form-distance-1}
\end{align}
where the second equality holds because of the following equation given by the product rule
\begin{align}
    Df(\rho)[X]=\tr(X\Omega(\rho))+\tr(\rho D\Omega(\rho)[X]).
\end{align}

Let $d = \dim \cH$. If $E_1,\ldots,E_m$ is a Hilbert--Schmidt orthonormal basis of $T_{\cM}$, where $m = d^2-k-1$, then Parseval's identity gives
\begin{align}
    &\operatorname{dist}_2^2\!\left(\Omega(\rho)-\nabla f(\rho),N_{\cM}\right) \notag\\
    &=\sum_{a=1}^m \bigl[ Df(\rho)[E_a]-\tr(E_a\Omega(\rho))\bigr]^2 \notag\\
    &=\sum_{a=1}^m \bigl[\tr(\rho D\Omega(\rho)[E_a])\bigr]^2.
    \label{eq:weak-form-distance-2}
\end{align}


We therefore derive an equivalent weak-form characterization that does not require the explicit formation of $\nabla f(\rho^*)$.

\begin{corollary}\label{corollary: weak-form conditions of equivalence}
    Let $\rho^*$ be an AB fixed point. Under the assumptions of Theorem~\ref{theorem: equivalence of fixed point}, the following conditions are equivalent:
    \begin{itemize}
        \item[(d1)] $\rho^*$ is the global minimizer;
        \item[(d2)] $Df(\rho^*)[X] = 0 \quad \forall X \in T_{\cM}$;
        \item[(d3)] $\tr(\rho^* D\Omega(\rho^*)[X]) = 0 \quad \forall X \in T_{\cM}$;
        \item[(d4)] $\tr(\rho^* D\Omega(\rho^*)[E_a]) = 0$, for $a = 1,\ldots,m$.
    \end{itemize}
\end{corollary}

\begin{proof}
    Since $\rho^*$ is an AB fixed point, Lemma~\ref{lemma: fixed point of AB} gives $\Omega(\rho^*)\in N_{\cM}$.
    Therefore,
    \begin{align}
        \tr(X\Omega(\rho^*)) = 0 \quad\forall X\in T_{\cM}.
    \end{align}
    The condition (c3) is equivalent to
    \begin{align}
        \tr(X\Omega(\rho^*)) = \tr(X\nabla f(\rho^*))\quad\forall X\in T_{\cM}.
    \end{align}
    We know that $\tr(X\nabla f(\rho^*)) = Df(\rho^*)[X]$. Therefore, condition (c3) is equivalent to 
    \begin{align}
        Df(\rho^*)[X] = 0 \quad\forall X\in T_{\cM}.
    \end{align}
    Then, by the product rule, we have
    \begin{align}
        Df(\rho)[X] = \tr(X\Omega(\rho))+ \tr(\rho D\Omega(\rho)[X]).
    \end{align}
    Thus, condition (d2) is equivalent to condition (d3).
    From \eqref{eq:weak-form-distance-2}, we know that if $\operatorname{dist}_2\!\left(\Omega(\rho)-\nabla f(\rho),N_{\cM}\right)=0$, then $\tr(\rho D\Omega(\rho)[E_a]) = 0$ for $a = 1,\ldots,m$. Thus, condition (d4) is equivalent to $\Omega(\rho^*) - \nabla f(\rho^*)\in N_{\cM}$.
\end{proof}

Moreover, the fixed-point optimality criterion admits the following extension to convex but possibly nonsmooth objective functions.
\begin{corollary}\label{theorem: subgradient}
Let $U\subseteq\cB(\cH)$ be an open convex set containing $\cM$, let $f:U\to\mathbb{R}$ be convex, and let $\rho^*\in\cM_{++}$ be an AB fixed point. Define the ambient subdifferential by
\begin{align}
\partial f(\rho)
:={}&
\bigl\{g\in\cB(\cH):
f(\tau)\geq f(\rho)+\langle g,\tau-\rho\rangle
\notag\\[-1mm]
&\hspace{42mm}\text{for every }\tau\in U\bigr\}.
\label{eq:ambient-subdifferential}
\end{align}
Then $\rho^*$ is a global minimizer of $f$ over $\cM$ if and only if
\begin{align}
\partial f(\rho^*)\cap N_{\cM}\neq\emptyset.
\end{align}
\end{corollary}

\subsection{A Posteriori Optimality Certificates}

The preceding results reveal the geometric property underlying the global optimality of the AB fixed point. However, condition (c3) is expressed in terms of the objective gradient $\nabla f(\rho^*)$. Directly computing $\nabla f(\rho^*)$ is undesirable in the applications for which the AB update is attractive. To verify whether a state produced by the AB iteration is globally optimal, we therefore propose an a posteriori certification scheme for the state produced by the AB iteration.

Corollary~\ref{corollary: weak-form conditions of equivalence} gives a zero-residual test: an AB fixed point is globally optimal exactly when every feasible first-order variation vanishes. For numerical purposes, one also needs to interpret a small but nonzero residual. The next theorem converts the distance from the objective gradient to $N_{\cM}$ into a rigorous upper bound on suboptimality.

\begin{theorem}\label{theorem:gap-certificate}
Let $f$ be convex and differentiable, let $f_*:=\min_{\sigma\in\cM}f(\sigma)$, and define
\begin{align}
\epsilon(\rho)&:=\dist_2(\nabla f(\rho),N_{\cM}),\\
R_{\cM}(\rho)&:=\sup_{\sigma\in\cM}\|\rho-\sigma\|_2.
\end{align}
For every $\rho\in\cM$,
\begin{align}
0\leq f(\rho)-f_*
\leq R_{\cM}(\rho)\epsilon(\rho)
\leq\sqrt{2}\,\epsilon(\rho).\label{eq:objective-gap}
\end{align}
\end{theorem}

\begin{proof}
Let $\rho^*\in\argmin_{\sigma\in\cM}f(\sigma)$ and decompose
$\nabla f(\rho)=P_T\nabla f(\rho)+P_N\nabla f(\rho)$.
For every $\sigma\in\cM$, convexity gives
\begin{align}
    f(\rho) - f(\sigma) \leq \tr[\nabla f(\rho)(\rho-\sigma)].
\end{align}
Since $\rho - \sigma \in T_{\cM}$,
\begin{align}
f(\rho)-f(\sigma)
&\leq\tr[\nabla f(\rho)(\rho-\sigma)] \notag\\
&=\tr[P_T\nabla f(\rho)(\rho-\sigma)]\notag\\
&\leq \|P_T\nabla f(\rho)\|_2\|\rho-\sigma\|_2 \notag\\
&\leq\epsilon(\rho)R_{\cM}(\rho).
\end{align}
Taking $\sigma = \rho^*$ gives
\begin{align}
    f(\rho) - f_* \leq \epsilon(\rho)R_{\cM}(\rho).
\end{align}
For density matrices $\rho$ and $\sigma$,
\begin{align}
\|\rho-\sigma\|_2^2=\tr\rho^2+\tr\sigma^2-2\tr(\rho\sigma)\leq2,
\end{align}
which proves the last inequality in \eqref{eq:objective-gap}. 
\end{proof}

\begin{remark}
    At an exact AB fixed point, $P_T\Omega(\rho)=0$, and hence $\dist_2(\nabla f(\rho),N_{\cM}) = \dist_2(\Omega(\rho) - \nabla f(\rho), N_{\cM})$.
\end{remark}

It remains to estimate $\epsilon(\rho)$ without explicitly calculating $\nabla f(\rho)$. Convexity supplies such estimates from one-dimensional secant slopes. Along each feasible line $t\mapsto \rho+tE_j$, the derivative at the origin lies between the backward and forward difference quotients.

Let $E_1,\ldots,E_m$ be an orthonormal basis of $T_{\cM}$ and let $\rho\in \cM_{++}$. Choose $h_j>0$ such that $\rho\pm h_jE_j\in \cM$. Since $\rho\succ 0$, it is sufficient to take
\begin{align}
    0<h_j< \frac{\lambda_{\min}(\rho)}{\|E_j\|_{\infty}}
\end{align}
to guarantee positivity of both perturbed states.

Defining
\begin{align}
    l_j:= \frac{f(\rho) - f(\rho-h_jE_j)}{h_j},\quad u_j :=  \frac{f(\rho+h_jE_j)-f(\rho)}{h_j},
\end{align}
and
\begin{align}
    \beta_j := \max \{|l_j|,|u_j|\}.
\end{align}

Therefore, the finite-difference method yields the suboptimality bound in the following theorem.

\begin{theorem}\label{theorem: finite difference certificate}
    If $f$ is convex and differentiable, then 
    \begin{align}
        l_j \leq Df(\rho)[E_j] \leq u_j \label{eq: bound for derivative}
    \end{align}
    and
    \begin{align}
        \epsilon(\rho) \leq  (\sum_{j=1}^m \beta_j^2)^{\frac12}.\label{eq:finite-difference-residual}
    \end{align}
    Consequently,
    \begin{align}
    0\leq f(\rho)-f_*
    \leq R_{\cM}(\rho)\left(\sum_{j=1}^m\beta_j^2\right)^{1/2}.
    \end{align}
\end{theorem}

\begin{proof}
    For fixed $j$, let $\phi_j(t) = f(\rho+tE_j)$. The function $\phi_j$ is convex. Hence its derivative at zero lies between the left and right secant slopes:
    \begin{align}
        \frac{\phi_j(0) - \phi_j(-h)}{h} \leq \phi'_j(0) \leq \frac{\phi_j(h) - \phi_j(0)}{h}.
    \end{align}
    Since $\phi'_j(0) = Df(\rho)[E_j]$, this proves~\eqref{eq: bound for derivative}. 
    We know that $|Df(\rho)[E_j]| \leq \beta_j$. Since 
    \begin{align}
        P_T\nabla f(\rho)=\sum_{j=1}^mDf(\rho)[E_j]E_j,
    \end{align}
    Parseval's identity yields \eqref{eq:finite-difference-residual}. The objective-gap bound follows from Theorem~\ref{theorem:gap-certificate}.
\end{proof}

\section{Examples}\label{sec: examples}

In this section, we illustrate the global optimality criterion and the a posteriori certificates using dephasing--depolarizing channel examples, both with and without additional linear constraints. The theoretical analysis identifies the global minimizer as the unique full-rank AB fixed point even though pathwise equivalence to MD fails. Numerical experiments illustrate convergence for the tested initializations and evaluate the certified objective-gap bounds.

\subsection{Formulation}
We consider the calculation of the quantum relative entropy of channels. 
Although this quantity has operational interpretations in channel discrimination and resource theories
\cite{cooney2016strong,wilde2020amortized,wang2019resource}, its numerical calculation is challenging because gradient-based solvers are hindered by complicated matrix functions \cite{liu2026posteriori}. 
Recently, numerical approaches based on semidefinite approximation can provide certified upper and lower bounds on this same channel divergence \cite{kossmann2025semidefinite}, and the Arimoto--Blahut algorithm provides another approach that is both efficient and certifiable \cite{liu2026posteriori}.
For quantum channels $\cN_{1,A\to B}$ and $\cN_{2,A\to B}$, where systems $A$ and $B$ are finite-dimensional Hilbert spaces, the quantum relative entropy between these two channels is defined as
\begin{equation}
D(\cN_1\|\cN_2):= \sup_{\rho_{AR}} D(\cN_{1,A\xrightarrow{}B}(\rho_{AR})\|\cN_{2,A\xrightarrow{}B}(\rho_{AR})).\label{eq: def1 of qre of channel}
\end{equation}

Using the Choi matrices of channels $\cN_1$ and $\cN_2$, we obtain an equivalent expression for \eqref{eq: def1 of qre of channel}:
\begin{align}\label{concave optimization of qre}
    D(\cN_1\|\cN_2) = 
    \sup_{\rho_A \in \cS(\cH_A)} 
    D(\sqrt{\rho_A} \Gamma_{AB}^{\cN_1}\sqrt{\rho_A}\|\sqrt{\rho_A}\Gamma_{AB}^{\cN_2}\sqrt{\rho_A}),
\end{align}
where $\Gamma_{AB}^{\cN_1}$ and $\Gamma_{AB}^{\cN_2}$ are the Choi matrices of quantum channels $\cN_1$ and $\cN_2$, respectively. It has been shown in \cite{khatri2020principles} that the optimization problem in \eqref{concave optimization of qre} is concave in $\rho_A$. Thus, calculation of the negative quantum relative entropy of channels can be written as the following convex optimization problem \cite{kossmann2025semidefinite}:
\begin{align}\label{SJA}
     \inf_{\rho_A \in \cS(\cH_A)} 
f(\rho_{A}),
\end{align}
where
\begin{align}
    f(\rho_{A}):=
    -D(\sqrt{\rho_A} \Gamma_{AB}^{\cN_1}\sqrt{\rho_A}\|\sqrt{\rho_A}\Gamma_{AB}^{\cN_2}\sqrt{\rho_A}).
    \label{eq: definition of f}
\end{align}
Define
\begin{align}  &\Omega_1(\rho_{A})\notag\\ &:= -\tr_B(\Gamma_{AB}^{\cN_1}\rho_A^{\frac{1}{2}} (\log \rho_A^{\frac{1}{2}} \Gamma_{AB}^{\cN_1}\rho_A^{\frac{1}{2}} - \log \rho_A^{\frac{1}{2}} \Gamma_{AB}^{\cN_2}\rho_A^{\frac{1}{2}})\rho_A^{-\frac{1}{2}}),
\notag\\
     &\Omega(\rho_{A}) :=\frac{\Omega_1(\rho_{A})+\Omega_1(\rho_{A})^{\dagger} }{2},\label{eq: definition of Omega}
\end{align}
It follows that 
\begin{align}
f(\rho_{A})
= \tr[\rho_{A}\Omega(\rho_{A})]. \label{ALA1}
\end{align}
With energy constraints, we can define the mixture family as follows:
\begin{align*}
\cM
:=\{\rho_{A}\in\cS(\cH_A):
\Tr \rho_{A}H_j =E_j,j=1,\ldots,k\}.
\end{align*}
Then the constrained problem is written as
\begin{align}\label{SJB}
     \inf_{\rho_A \in \cM} 
f(\rho_{A}).
\end{align}


\subsection{Theoretical Analysis}
To conduct a concrete analysis and numerical experiments, we consider the channel relative entropy between a dephasing channel and a depolarizing channel. The depolarizing channel $\cD_p$ with parameter $p$ is defined by
$ \cD_p(\rho) = (1-p)\rho + p\frac{I}{2}$, and
the dephasing channel is defined as 
$    \cD_{deph}(\rho) = q\rho + (1-q)\sigma_z\rho\sigma_z$,
where $\sigma_z$ denotes the Pauli Z matrix. The dimensions of systems $A$ and $B$ are $d_A=d_B=2$. The parameter of the dephasing channel $\cD_{deph}$ is chosen as $q = 0.4$, and the depolarizing parameter is chosen as $p = 0.05$ in the following analysis and numerical experiments.

We first investigate the case without linear constraints. In this case, both the dephasing and depolarizing channels are Pauli channels. Thus, these two channels are teleportation-covariant~\cite{pirandola2017fundamental}. According to the teleportation-covariance property and the data-processing inequality,
the maximization in \eqref{eq: def1 of qre of channel} for these two channels is attained at the maximally mixed state, which yields the exact value of $D(\cD_{deph}\|\cD_p)$. In this scenario, the normal space is $N_{\cM} = \{cI:c\in \mathbb{R}\}$.

Then we have the following theorem:
\begin{theorem}\label{theorem: convergence without constraint}
     The maximally mixed state is the unique full-rank fixed point of the AB algorithm in the above unconstrained problem.
\end{theorem}
We know that the problem \eqref{SJA} is a convex optimization problem and that the global minimizer is the maximally mixed state. Theorem~\ref{theorem: convergence without constraint} identifies the maximally mixed state as the unique full-rank fixed point of the AB algorithm. Consequently, if the AB algorithm converges to a full-rank limit, then its limit must be the maximally mixed state, which is the global minimizer of the problem.

To show that the AB algorithm is not the same as the MD algorithm in this example, we propose the following theorem.
\begin{theorem}\label{theorem: no equivalence without constraint}
    $\Omega(\rho) - \nabla f(\rho)\in N_{\cM} $ does not hold universally over the feasible set $\cS(\cH_A)$, where $N_{\cM} = \{cI:c\in \mathbb{R}\}$.
\end{theorem}
We provide a concrete counterexample showing that $\Omega(\rho) - \nabla f(\rho) \notin N_{\cM}$ in the following numerical subsection. The detailed calculation of $\nabla f(\rho)$ can be found in Appendix~\ref{appendix: calculation of derivative}.

Combining Theorem~\ref{theorem: equivalence of each step} and Theorem~\ref{theorem: no equivalence without constraint}, we know that $\cT_{AB}(\rho) = \cT_{MD}(\rho)$ does not always hold for $\rho\in \cS(\cH_A)$. Therefore, the AB algorithm can behave differently from the MD algorithm even when they start from the same initial point.

We next consider the scenario with a linear constraint. For convenience in the theoretical analysis, we assume that only one linear constraint is imposed and that the Hermitian operator $H$ is chosen to be a Pauli matrix; i.e., the linear constraint is $\tr\rho H = E$, where $H\in \{\sigma_x,\sigma_y,\sigma_z\}$.
Thus, in this case, the mixture family is defined as $\cM = \{\rho_A\in\cS(\cH_A): \tr\rho_A H = E\}$,
 and the normal space is $N_{\cM} = \text{span}\{I,H\}$.
\begin{proposition}\label{proposition: global minimum with constraint}
     $\hat{\rho} = \frac12(I+EH)$ is a global minimizer when the feasible set is $\cM = \{\rho_A\in\cS(\cH_A): \tr\rho_A H = E\}$ and $H\in \{\sigma_x,\sigma_y,\sigma_z\} $.
\end{proposition}
 
\begin{theorem}\label{theorem: convergence with constraint}
    Let \(E\in(-1,1)\). The state $\hat{\rho} = \frac12(I+EH)$ is the unique full-rank fixed point of the AB algorithm in the above constrained case.
\end{theorem}
Therefore, combining Proposition~\ref{proposition: global minimum with constraint} and Theorem~\ref{theorem: convergence with constraint}, we conclude that $\hat{\rho}= \frac12(I+EH)$ is both a global minimizer and the unique full-rank fixed point of the AB algorithm on the constrained feasible set. 
Consequently, if an AB trajectory converges to a full-rank limit, then continuity of the AB map implies that the limit is an AB fixed point, and Theorem~\ref{theorem: convergence with constraint} forces it to be $\hat{\rho}$.
As in the unconstrained case, we can also prove that the pathwise equivalence between the AB algorithm and the MD algorithm does not hold in general.
\begin{theorem}\label{theorem: no equivalence with constraint}
    $\Omega(\rho) - \nabla f(\rho)\in N_{\cM} $ does not hold universally over the feasible set $\cM$, where $N_{\cM} = \text{span}\{I,H\}$.
\end{theorem}

The theoretical results identify the unique full-rank AB fixed point with the global minimizer, while the following numerical experiments show convergence to this state for the tested initializations.

\subsection{Numerical Comparison of AB and MD}
In the numerical experiments, the parameter of the dephasing channel $\cD_{deph}$ is fixed at $q = 0.4$, and the depolarizing parameter is chosen as $p = 0.05$. In our example, we choose the step size $\gamma= 1$.
The numerical verification of condition (a1) for the reported trajectories is summarized in Remark~\ref{remark:numerical-a1}.
Next, we consider the case without linear constraints. The initial point is chosen as the state 
\begin{equation}
\rho^{(0)} = \begin{pmatrix}
    \frac{3}{8}&0\\
    0&\frac{5}{8}
\end{pmatrix}.
\end{equation}
We then apply the AB algorithm and mirror descent with the initial point $\rho^{(0)}$, and the result is shown in Fig.~\ref{fig: exp without constraint}. 
\begin{figure}[htbp]
    \centering
    \includegraphics[width = 0.9\columnwidth]{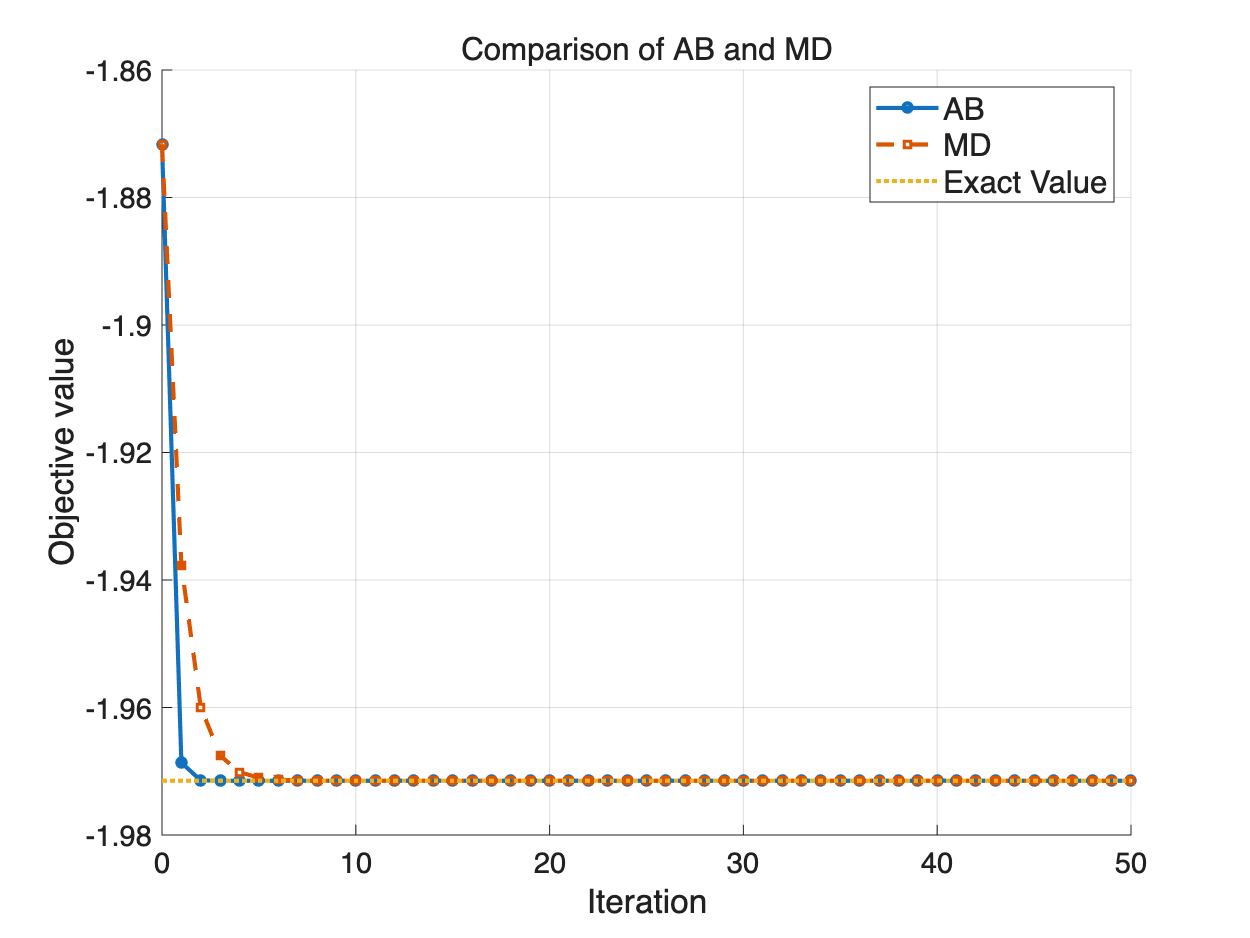}
    \caption{\textbf{Comparison of AB and MD iterative behavior in the unconstrained case}. The vertical axis shows the value of $f(\rho)$. The horizontal axis shows the number of iterations. The blue line represents the iterative behavior of the AB algorithm. The orange line represents the iterative behavior of the MD algorithm. The yellow line represents the theoretical value of $-D(\cD_{deph}\|\cD_p)$.}
    \label{fig: exp without constraint}
\end{figure}
Fig.~\ref{fig: exp without constraint} shows that the AB and MD iterates numerically approach the same global minimizer while following different trajectories and converging at different rates.
At the point $\rho^{(0)}$, $\nabla f(\rho^{(0)})$ and $\Omega(\rho^{(0)})$ are calculated as follows:
\begin{align}
    \nabla f(\rho^{(0)}) &= \begin{pmatrix}
        -0.4023 &0\\
        0& 0.4023
    \end{pmatrix},\\
    \Omega(\rho^{(0)}) &= \begin{pmatrix}
        -1.2023 & 0\\
        0& -0.7759
    \end{pmatrix}.
\end{align}
Thus, we have
\begin{align}\label{eq: nonequivalence for unconstrain-1}
    &\Omega(\rho^{(0)}) - \nabla f(\rho^{(0)}) = \begin{pmatrix}
        -0.8000&0\\
        0& -1.1782
    \end{pmatrix},
\end{align}
and
\begin{align}\label{eq: nonequivalence for unconstrain-2}
    \|P_T(\Omega(\rho^{(0)}) - \nabla f(\rho^{(0)}))\|_2 = 0.2674.
\end{align}
Since we consider the problem without constraints, any element belonging to $N_{\cM}$ must be proportional to the identity. Therefore, from \eqref{eq: nonequivalence for unconstrain-1} and \eqref{eq: nonequivalence for unconstrain-2}, we can get that $\Omega(\rho^{(0)}) - \nabla f(\rho^{(0)}) \notin N_{\cM}$. By Theorem~\ref{theorem: equivalence of each step}, this result indicates that $\cT_{AB}(\rho^{(0)})\neq \cT_{MD}(\rho^{(0)})$, which agrees with the numerical evidence shown in Fig.~\ref{fig: exp without constraint}.

Then we show another example of the above problem with linear constraints. Here the linear constraint is described by $\tr(\rho \sigma_x) = 0.4$. The feasible set is $\cM = \{\rho\in \cS(\cH_A): \tr(\rho \sigma_x) = 0.4\}$. To ensure that the algorithm starts from a point in $\cM$, the initial point is chosen as the state
\begin{align}
    \rho^{(0)} = \begin{pmatrix}
        \frac38 & 0.2+0.2i\\
        0.2-0.2i &\frac58
    \end{pmatrix}.
\end{align}
It is easy to verify that $\rho^{(0)}$ satisfies the condition $\tr(\rho^{(0)}\sigma_x) = 0.4$.
Starting from $\rho^{(0)}$, we compare the iterative behavior of the AB algorithm and the MD algorithm; the result is shown in Fig.~\ref{fig: exp with constraint}. Similar to the unconstrained case, both numerical trajectories of the AB algorithm and the MD algorithm approach the known global minimizer while their paths are different. We then calculate $\nabla f(\rho^{(0)})$ and $\Omega(\rho^{(0)})$.
\begin{align}
    \nabla f(\rho^{(0)}) &= \begin{pmatrix}
        -0.4345 &0.1496+0.1495i\\
        0.1496-0.1495i& 0.4345
    \end{pmatrix},\\
    \Omega(\rho^{(0)}) &= \begin{pmatrix}
        -1.2998& 0.1452+0.1452i\\
        0.1452-0.1452i& -0.8144
    \end{pmatrix}.
\end{align}
Thus
\begin{align}\label{eq: nonequivalence for constrain-1}
    &\Omega(\rho^{(0)}) - \nabla f(\rho^{(0)}) \notag \\
    = &\begin{pmatrix}
        -0.8652&-0.0042-0.0042i\\
        -0.0042+0.0042i& -1.2489
    \end{pmatrix},
\end{align}
and
\begin{align}\label{eq: nonequivalence for constrain-2}
    \|P_T(\Omega(\rho^{(0)}) - \nabla f(\rho^{(0)}))\|_2 = 0.2714.
\end{align}
From \eqref{eq: nonequivalence for constrain-1} and \eqref{eq: nonequivalence for constrain-2}, we can see that $\Omega(\rho^{(0)}) - \nabla f(\rho^{(0)})\notin N_{\cM} = \text{span}\{I,\sigma_x\}$. According to Theorem~\ref{theorem: equivalence of each step}, we know that $\cT_{AB}(\rho^{(0)})\neq \cT_{MD}(\rho^{(0)})$, which agrees with the numerical evidence shown in Fig.~\ref{fig: exp with constraint}.

\begin{figure}[htbp]
    \centering
    \includegraphics[width = 0.9\columnwidth]{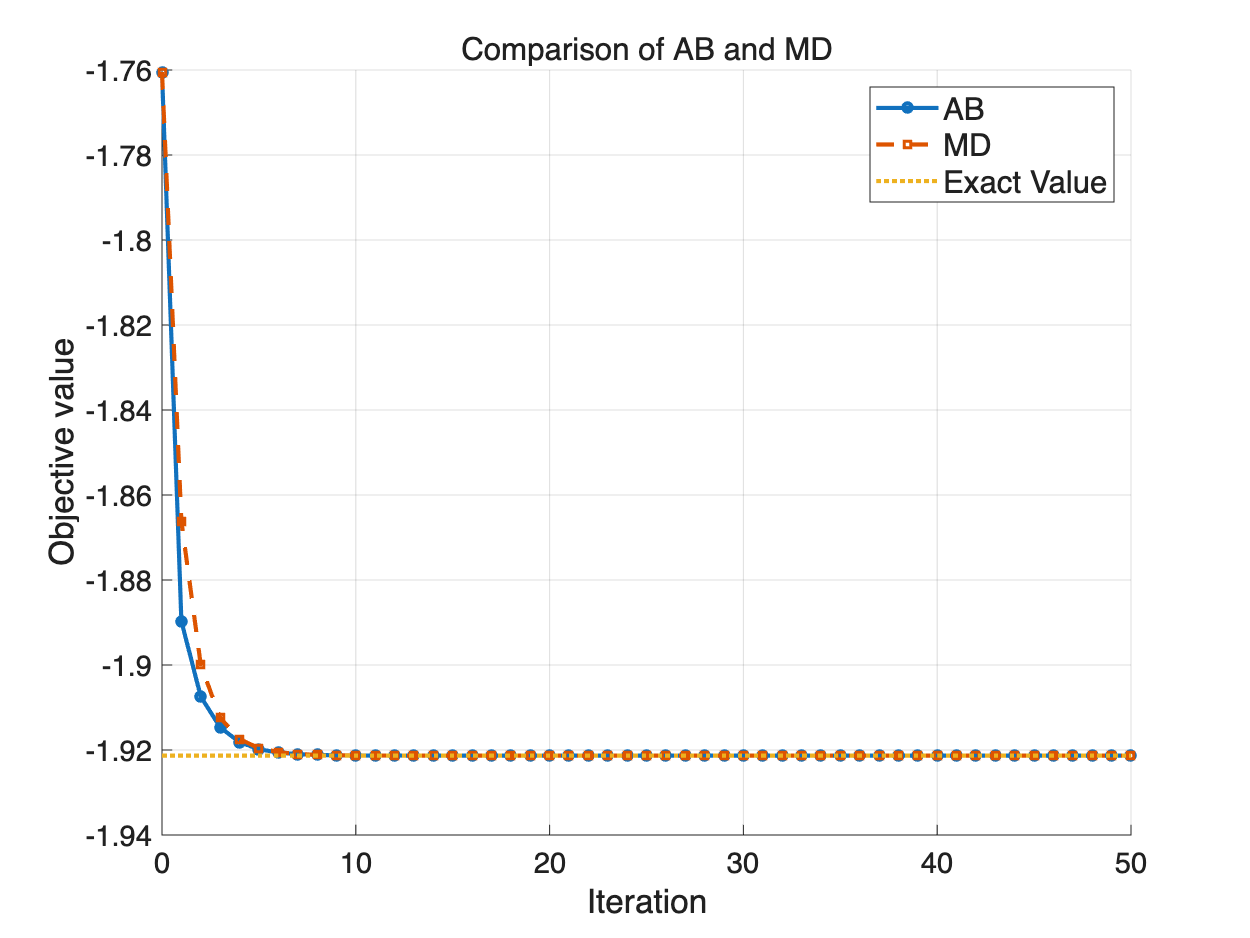}
    \caption{\textbf{Comparison of AB and MD iterative behavior in the constrained case}. The vertical axis shows the value of $f(\rho)$. The horizontal axis shows the number of iterations. The blue line represents the iterative behavior of the AB algorithm. The orange line represents the iterative behavior of the MD algorithm. The yellow line represents the theoretical value of $-D(\cD_{deph}\|\cD_p)$.}
    \label{fig: exp with constraint}
\end{figure}

\begin{remark}[Numerical verification of condition (a1)]
\label{remark:numerical-a1}
The choice $\gamma=1$ is consistent with the a posteriori
study in \cite{liu2026posteriori}, where the trajectory inequality corresponding to condition (a1) was numerically examined for the same dephasing--depolarizing channel family.

Since condition (a1) depends on the particular trajectory, we also evaluate it directly for the two trajectories reported above. For every iteration satisfying
\[
D(\rho^{(t+1)}\Vert\rho^{(t)})>10^{-10},
\]
we calculate
\begin{align}
r_t:=
\frac{
D_{\Omega}(\rho^{(t+1)}\Vert\rho^{(t)})
}{
D(\rho^{(t+1)}\Vert\rho^{(t)})
}.
\end{align}
The maximum ratios along the reported trajectories are
\begin{align}
\max_t r_t
&=-0.1669
&&\text{in the unconstrained case},\\
\max_t r_t
&=0.4356
&&\text{in the constrained case}.
\end{align}
Both values are below $\gamma=1$. Therefore, condition (a1) is satisfied along all numerically resolved iterations reported here.
We emphasize that $D_{\Omega}$ is not necessarily nonnegative. Therefore, a negative ratio reflects the fact that the stronger nonnegativity condition (a2) need not hold.
\end{remark}

\subsection{A Posteriori Certification}

The preceding experiments demonstrate that the AB and MD algorithms may follow different trajectories while approaching the same global solution. We now apply the a posteriori certificate developed in Section~\ref{sec: equivalence} to quantify the accuracy of the states returned by the AB iteration. The certificate uses only objective evaluations along feasible tangent directions and does not require the explicit construction of $\nabla f(\rho)$.

For each AB iterate $\rho^{(t)}$, let $\{E_a\}_{a=1}^{m}$ be a Hilbert--Schmidt orthonormal basis of $T_{\mathcal M}$. For a feasible step size $h_a>0$, define
\begin{align}
\ell_{a,t}&=\frac{f(\rho^{(t)})-f(\rho^{(t)}-h_aE_a)}{h_a},\\
u_{a,t}&=\frac{f(\rho^{(t)}+h_aE_a)-f(\rho^{(t)})}{h_a},
\end{align}
and
\begin{equation}
\widehat\epsilon_t=\left(\sum_{a=1}^{m}\max\{|\ell_{a,t}|,|u_{a,t}|\}^2\right)^{1/2}.
\end{equation}
By Theorem~\ref{theorem: finite difference certificate},
\begin{equation}
0\leq f(\rho^{(t)})-f_*\leq R_{\mathcal M}(\rho^{(t)})\widehat\epsilon_t \leq \sqrt{2}\widehat\epsilon_t.
\end{equation}
Because the exact minimizers are known analytically in the present examples, we may also compare the certified bound with the true objective gap. The exact solution is used only for validation and is not an input to the certificate.

For the unconstrained qubit problem, the tangent space is
\begin{align}
    T_{\cM} = \{X =X^{\dagger}: \tr X = 0\},
\end{align}
We use the Hilbert--Schmidt orthonormal basis
\begin{align}
    E_x = \frac{\sigma_x}{\sqrt{2}}, \quad E_y = \frac{\sigma_y}{\sqrt{2}},\quad E_z = \frac{\sigma_z}{\sqrt{2}}.
\end{align}
For the constrained problem where $\cM = \{\rho\in \cS(\cH_A): \tr(\rho \sigma_x) = 0.4\}$, we use
\begin{align}
    E_y = \frac{\sigma_y}{\sqrt{2}},\quad E_z = \frac{\sigma_z}{\sqrt{2}}.
\end{align}

For each state along the AB trajectory, we evaluate the true objective gap $f(\rho^{(t)})-f_*$ and the certified suboptimality bound $\sqrt{2}\widehat\epsilon_t$. For the unconstrained problem, we use $h\in\{10^{-6},10^{-8},10^{-10},10^{-12}\}$, whereas for the constrained problem, we use $h\in\{10^{-4},10^{-6},10^{-8},10^{-10}\}$. The results are shown in Fig.~\ref{fig: certification without constraint} and Fig.~\ref{fig: certification with constraint}, respectively. In both cases, all certified bounds remain above the true objective gap throughout the AB trajectory, numerically confirming the validity of the a posteriori certificate. Within the tested range, smaller values of $h$ yield tighter bounds and lower numerical floors. In the unconstrained case, the true gap reaches approximately machine precision, while in the constrained case it stabilizes near $10^{-10}$ because of numerical error. Nevertheless, the certificate continues to provide a valid and informative upper bound in both settings.

\begin{figure}[htbp]
    \centering
    \includegraphics[width = \columnwidth]{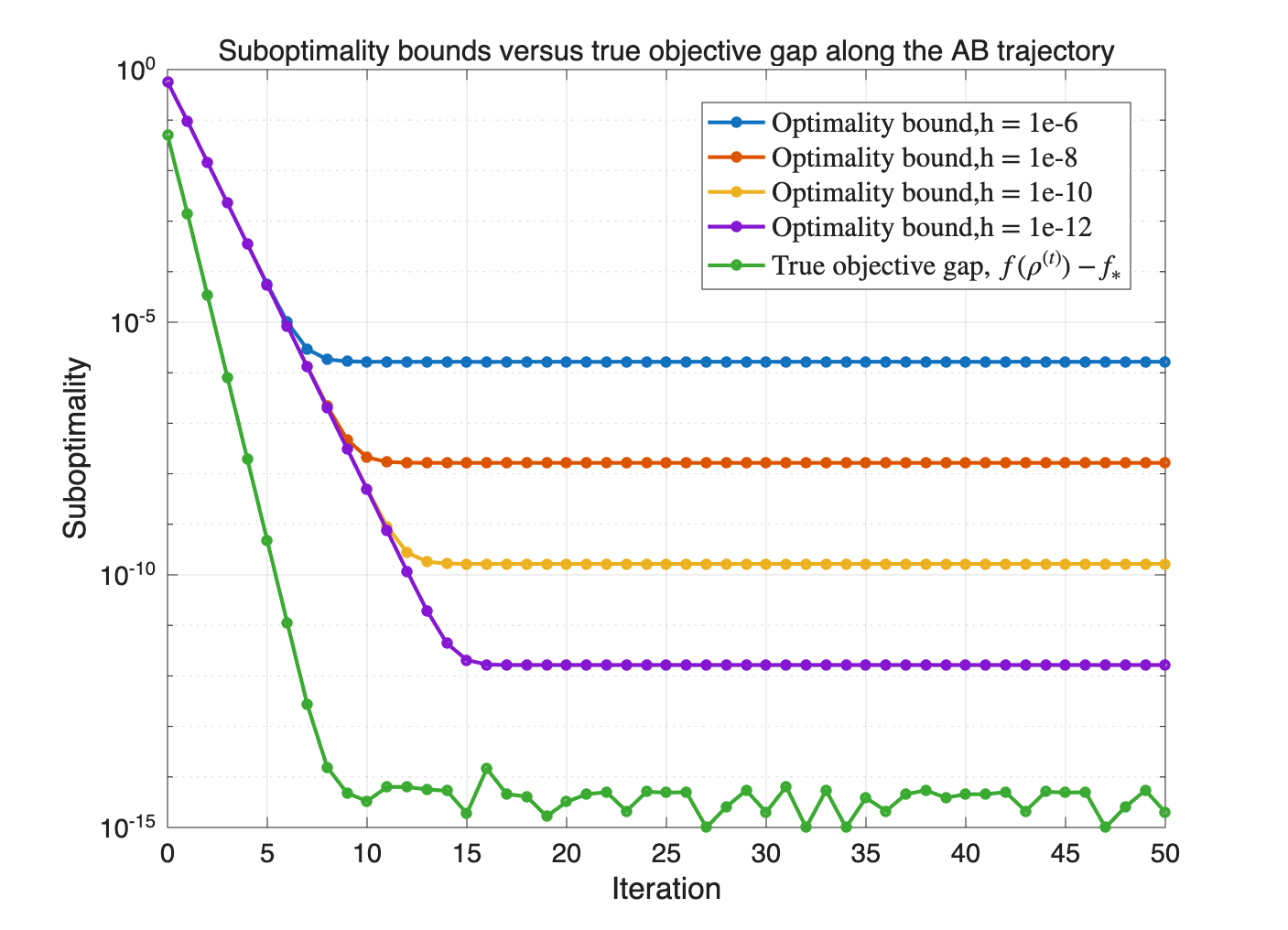}
    \caption{\textbf{Comparison of certified suboptimality bounds and the true objective gap along the AB trajectory in the unconstrained case}. The vertical axis shows the suboptimality on a logarithmic scale. The horizontal axis shows the number of iterations. The blue, orange, yellow, and purple lines represent the certified bounds obtained with $h=10^{-6},10^{-8},10^{-10}$, and $10^{-12}$, respectively. The green line represents the true objective gap $f(\rho^{(t)})-f_*$.}
    \label{fig: certification without constraint}
\end{figure}

\begin{figure}[htbp]
    \centering
    \includegraphics[width = 0.9\columnwidth]{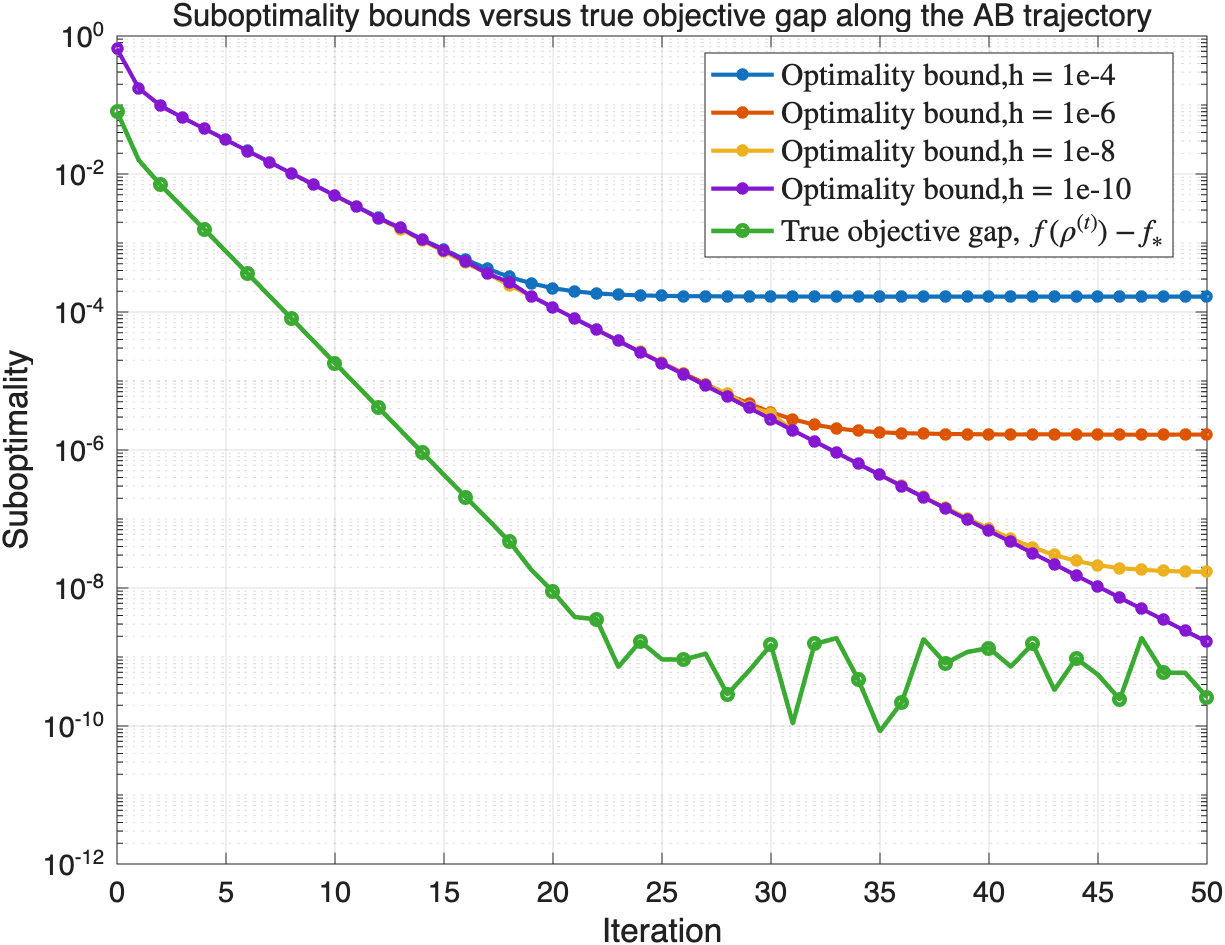}
    \caption{\textbf{Comparison of certified suboptimality bounds and the true objective gap along the AB trajectory in the constrained case}. The vertical axis shows the suboptimality on a logarithmic scale. The horizontal axis shows the number of iterations. The blue, orange, yellow, and purple lines represent the certified bounds obtained with $h=10^{-4},10^{-6},10^{-8}$, and $10^{-10}$, respectively. The green line represents the true objective gap $f(\rho^{(t)})-f_*$.}
    \label{fig: certification with constraint}
\end{figure}

These experiments illustrate the practical role of the a posteriori certificate. It both assesses whether an AB iterate is approximately optimal and provides a computable, conservative upper bound without using the unknown optimal value $f_*$. The tightness of the bound depends on the finite-difference step size $h$. Decreasing $h$ generally yields a tighter bound, whereas an excessively small value may amplify floating-point cancellation.

\section{Counterexample}\label{sec: counterexample}
We now present a numerical example showing that the convergence of the AB iteration, even with monotonic decrease of a convex objective, does not by itself guarantee global optimality. Consider the qubit amplitude damping channel
\begin{align}
    \cD_{ad}(\rho) = K_0\rho K_0^{\dagger} + K_1\rho K_1^{\dagger},
\end{align}
with Kraus operators
\begin{align}
    K_0 = \begin{pmatrix}
1&0\\
0&\sqrt{1-\eta}
\end{pmatrix},
\qquad
K_1=
\begin{pmatrix}
0&\sqrt{\eta}\\
0&0
\end{pmatrix}.
\end{align}
We minimize the negative channel relative entropy between the amplitude damping channel and the depolarizing channel $D(\cD_{ad}\|\cD_p)$ over all qubit density operators.

In this example, we choose $\eta = 0.5$, $p = 0.05$, and $\gamma = 1$ and initialize the AB algorithm and the MD algorithm at the maximally mixed state. Fig.~\ref{fig: counter example} shows the objective values produced by the two algorithms during the first 50 iterations. Both objective sequences decrease monotonically and become numerically stable. However, they converge to different limiting values. The AB iteration converges to a state whose objective value is strictly larger than the numerical result obtained by the MD algorithm. Thus, this example demonstrates that the monotonic convergence of the AB iteration is insufficient to establish global optimality.

\begin{figure}[htbp]
    \centering
    \includegraphics[width = \columnwidth]{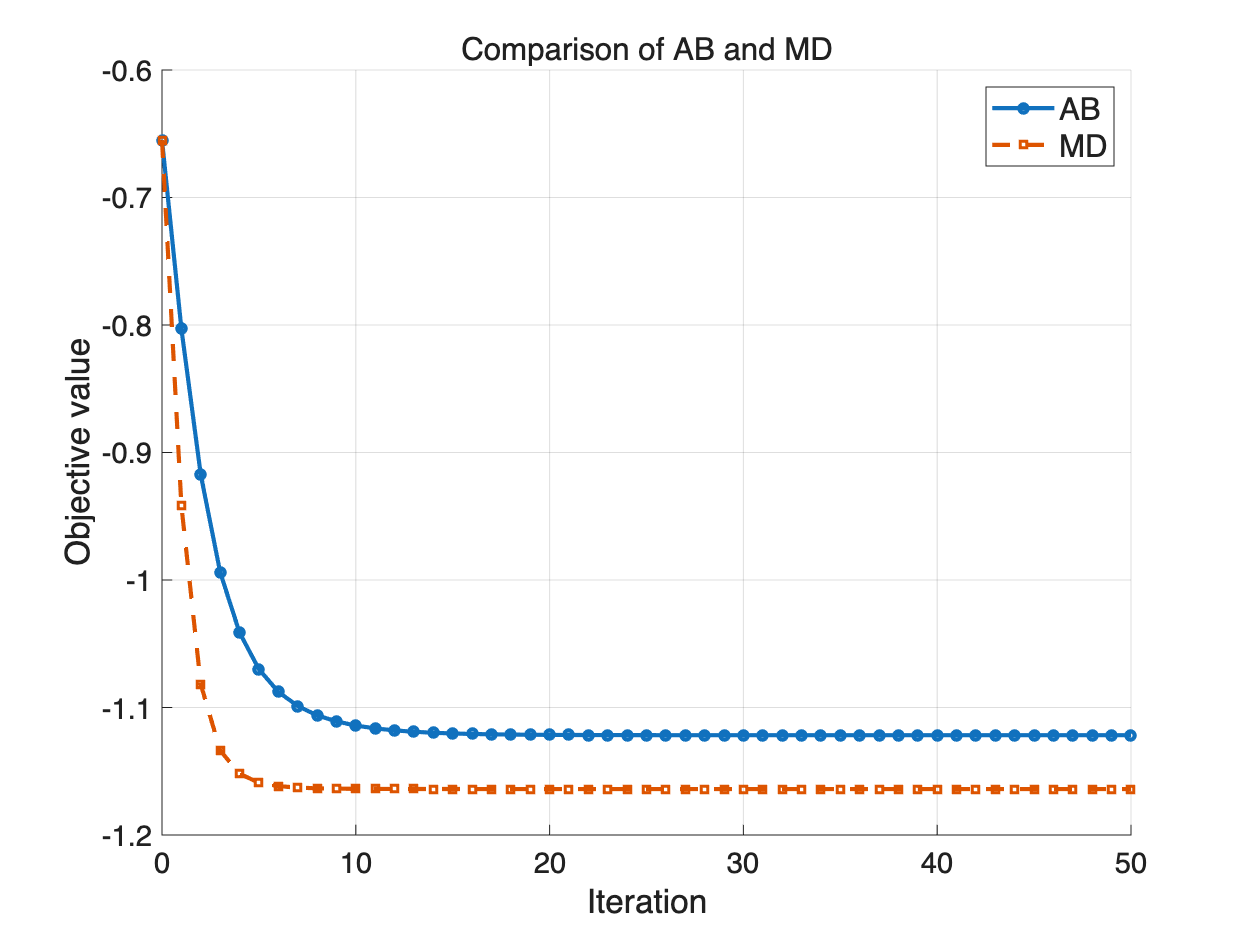}
    \caption{\textbf{Comparison of AB and MD iterative behavior associated with $D(\cD_{ad}\|\cD_p)$}. The vertical axis shows the value of $f(\rho)$. The horizontal axis shows the number of iterations. The blue line represents the iterative behavior of the AB algorithm. The orange line represents the iterative behavior of the MD algorithm.}
    \label{fig: counter example}
\end{figure}

Let $\rho^{(50)}$ denote the AB iterate after 50 iterations. At $\rho^{(50)}$, the fixed-point residual is
\begin{align}
\left\|
\mathcal T_{\mathrm{AB}}(\rho^{(50)})-\rho^{(50)}
\right\|_1 \approx 7.8\times 10^{-9},
\end{align}
and the objective-value change satisfies
\begin{align}
|f(\rho^{(50)})-f(\rho^{(49)})| \approx 3.4\times 10^{-9}.
\end{align}
We therefore regard $\rho^{(50)}$ as a numerically converged AB fixed point.
Then we apply the finite-difference certification to the terminal AB state $\rho^{(50)}$. For the unconstrained qubit problem, we use the Hilbert--Schmidt orthonormal basis
\begin{align}
E_x=\frac{\sigma_x}{\sqrt{2}},\quad
E_y=\frac{\sigma_y}{\sqrt{2}},\quad
E_z=\frac{\sigma_z}{\sqrt{2}}.
\end{align}
For each direction $E_j$, we compute the finite-difference interval $[\ell_j,u_j]$ defined in Theorem~\ref{theorem: finite difference certificate}. If $0\notin[\ell_j,u_j]$ for at least one $j$, then $Df(\rho^{(50)})[E_j]\neq0$. Consequently, $\nabla f(\rho^{(50)})\notin N_{\cM}$, and Corollary~\ref{corollary: weak-form conditions of equivalence} certifies that $\rho^{(50)}$ is not globally optimal.

In the present example, the interval associated with $E_z$ excludes zero, whereas those associated with $E_x$ and $E_y$ are numerically close to zero. We take $h=10^{-4}$ and calculate $l_{z,50}$ and $u_{z,50}$. The numerical result shows that
\begin{align}
    0.6127< l_{z,50}\leq Df(\rho^{(50)})[E_z] \leq u_{z,50} < 0.6129.
\end{align}
This provides a finite-difference numerical certificate, requiring only objective evaluations, that the AB fixed point violates the condition in Corollary~\ref{corollary: weak-form conditions of equivalence} and is therefore not globally optimal, which agrees with the numerical evidence shown in Fig.~\ref{fig: counter example}.

Moreover, direct calculation gives
\begin{align}
    \nabla f(\rho^{(50)}) &= \begin{pmatrix}
        -0.4333 &0\\
        0& 0.4333
    \end{pmatrix},\\
    \Omega(\rho^{(50)}) &= \begin{pmatrix}
        -1.1217 & 0\\
        0& -1.1217
    \end{pmatrix}.
\end{align}
Since $\Omega(\rho^{(50)})$ is proportional to the identity, it belongs to the normal space $N_{\cM}=\{cI:c\in\mathbb{R}\}$, consistent with $\rho^{(50)}$ being an AB fixed point. In contrast, $\nabla f(\rho^{(50)})$ is not proportional to the identity. Equivalently,
\begin{align}
\Omega(\rho^{(50)})-\nabla f(\rho^{(50)}) \notin N_{\cM}.
\end{align}
By Theorem~\ref{theorem: equivalence of fixed point}, this calculation directly confirms the conclusion of the numerical certification above.

\section{Discussion}\label{sec: discussion}


In this paper, we have established a necessary and sufficient condition for the global optimality of full-rank fixed points of generalized quantum AB algorithms with convex differentiable objectives and linear constraints. An AB fixed point is globally optimal precisely when the AB direction and the objective gradient differ by an element of the constraint normal space.  It also clarifies the role of mirror descent: agreement of the two algorithms along a trajectory requires compatibility at every common iterate, whereas the optimality of an AB fixed point depends only on compatibility at that point.


This distinction also changes how an AB trajectory should be interpreted. A convergent AB trajectory may have a globally optimal limit even when it differs from the corresponding MD trajectory. Conversely, monotonic decrease and apparent numerical convergence do not by themselves certify global optimality. 
The complementary examples in this paper illustrate both possibilities and show that the geometry of the candidate terminal state is decisive. The terminal state should therefore be assessed directly. The directional-derivative and finite-difference certificates developed here provide such an a posteriori assessment and a computable upper bound on the objective gap without explicitly computing the objective gradient.

These results characterize optimality at a full-rank AB fixed point and provide certificates for finite iterates. Establishing convergence to a globally optimal fixed point requires additional analysis. A central direction for future work is to develop checkable conditions that ensure such convergence without requiring pathwise equivalence to MD. It would also be useful to explain theoretically the differing empirical convergence behavior observed in our examples. Other directions include extending the analysis to rank-deficient boundary fixed points and feasible sets beyond linear equality constraints. In higher-dimensional problems, the cost of probing a full tangent-space basis and the numerical choice of finite-difference step sizes motivate more scalable and adaptive certification procedures.

\appendices


\section{Proof of Corollary~\ref{theorem: subgradient}}
First, suppose that $g_*\in\partial f(\rho^*)\cap N_{\cM}$. For every $\sigma\in\cM$, the difference $\sigma-\rho^*$ lies in $T_{\cM}$. Since $N_{\cM}=T_{\cM}^{\perp}$, we have $\langle g_*,\sigma-\rho^*\rangle=0$. The subgradient inequality therefore gives
\begin{align}
f(\sigma)
\geq f(\rho^*)+\langle g_*,\sigma-\rho^*\rangle
=f(\rho^*).
\end{align}
Thus, $\rho^*$ is a global minimizer of $f$ over $\cM$.

Conversely, suppose that $\rho^*$ is a global minimizer. Since $\rho^*\succ0$, the positive-semidefinite constraint is inactive at $\rho^*$. The standard first-order optimality condition for a convex problem with affine equality constraints therefore gives a subgradient $g_*\in\partial f(\rho^*)$ and real numbers $\lambda_0,\lambda_1,\ldots,\lambda_k$ such that
\begin{align}
g_*+\lambda_0 I+\sum_{j=1}^k\lambda_jH_j=0.
\end{align}
Here $I$ corresponds to the normalization constraint and each $H_j$ corresponds to one of the linear moment constraints. It follows that
\begin{align}
g_* \in\operatorname{span}\{I,H_1,\ldots,H_k\}=N_{\cM}.
\end{align}
Hence $\partial f(\rho^*)\cap N_{\cM}\neq\emptyset$, which proves the converse.

\section{Proof of Proposition~\ref{proposition: global minimum with constraint}}
Without loss of generality, we assume that $H=\sigma_x$; the following proof can be generalized to $H=\sigma_y$ and $H=\sigma_z$.
We define the following map
\begin{align}
    \Pi_{x}(\rho) = \frac12(\rho+\sigma_x\rho\sigma_x).
\end{align}
If 
\begin{align}
    \rho = \frac12(I+a_x\sigma_x+a_y\sigma_y+a_z\sigma_z),
\end{align}
then
\begin{align}
    \sigma_x\rho\sigma_x = \frac12(I+a_x\sigma_x -a_y\sigma_y-a_z\sigma_z),
\end{align}
and therefore
\begin{align}
    \Pi_x(\rho) = \frac{1}{2}(I+a_x\sigma_x).
\end{align}
Since $\rho\in\cM$ and satisfies the constraint $\tr\rho H=E$,
\begin{align}
    \Pi_x(\rho) = \frac12(I+EH).
\end{align}
We next prove that
\begin{align}
    f(\Pi_x(\rho))\leq f(\rho),\quad\forall \rho\in\cM.
\end{align}
Since both channels are Pauli-covariant, both Choi matrices satisfy the covariance relation
\begin{align}
    (\sigma_x\otimes\sigma_x) \Gamma^{\cN_1}(\sigma_x\otimes\sigma_x) & = \Gamma^{\cN_1},\\
    (\sigma_x\otimes\sigma_x) \Gamma^{\cN_2}(\sigma_x\otimes\sigma_x) & = \Gamma^{\cN_2}.
\end{align}
Therefore,
\begin{align}
\rho_x^{\frac12}\Gamma^{\cN_1}\rho_x^{\frac12} = (\sigma_x\otimes\sigma_x)(\rho^{\frac12}\Gamma^{\cN_1}\rho^{\frac12})(\sigma_x\otimes\sigma_x),
\end{align}
where $\rho_x= \sigma_x\rho\sigma_x$. The above property also holds for $\Gamma^{\cN_2}$.
Since quantum relative entropy is invariant under isometries, we have
\begin{align}
    f(\sigma_x\rho\sigma_x) = f(\rho).
\end{align}
Therefore, by the convexity of $f$, we have
\begin{align}
    f( \frac{1}{2}(I+EH)) = &f(\Pi_x(\rho)) \\ 
    = &f(\frac12\rho+\frac12\sigma_x\rho\sigma_x) \notag\\
    \leq &\frac12 f(\rho)+\frac12 f(\sigma_x\rho\sigma_x) = f(\rho),
\end{align}
for any $\rho\in\cM$.

Therefore, $\hat{\rho} = \frac{1}{2}(I+EH)$ is a global minimizer.

\section{Proof of Theorem~\ref{theorem: convergence without constraint}}
\label{proof of theorem convergence without constraint}

This appendix proves that, for the dephasing--depolarizing example with
\(q=0.4\) and \(p=0.05\), the maximally mixed state is the unique
full-rank fixed point of the AB map.  The proof has three parts.  First,
we reduce the fixed-point problem to proving the strict positivity of
the scalar function \(T(\rho)\).  Second, we derive a closed-form
expression for \(T(\rho)\) by differentiating along radial lines of the
Bloch ball.  Third, we prove that its coefficients are strictly
positive for the chosen channel parameters.

\subsection{Scalar Reduction of the Fixed-Point Problem}

We begin with a general uniqueness lemma that will also be used in
Appendix~\ref{proof of theorem convergence with constraint}.

\begin{lemma}[Scalar reduction of fixed-point uniqueness]
\label{lemma: scalar fixed-point certificate}
Let \(\widehat\rho\in\cM_{++}\) be an AB fixed point.  If
\begin{equation}
\Tr[(\rho-\widehat\rho)\Omega(\rho)]>0
\qquad
\text{for every }\rho\in\cM_{++}\setminus\{\widehat\rho\},
\label{eq: scalar fixed-point certificate}
\end{equation}
then \(\widehat\rho\) is the unique AB fixed point in \(\cM_{++}\).
\end{lemma}

\begin{proof}
Suppose that \(\rho\in\cM_{++}\) is an AB fixed point.
Lemma~\ref{lemma: fixed point of AB} gives
\(\Omega(\rho)\in N_{\cM}\).  Because both \(\rho\) and
\(\widehat\rho\) belong to \(\cM\), their difference lies in
\(T_{\cM}\).  Hence
\begin{equation}
\Tr[(\rho-\widehat\rho)\Omega(\rho)]=0.
\end{equation}
Condition~\eqref{eq: scalar fixed-point certificate} therefore rules
out every fixed point other than \(\widehat\rho\).
\end{proof}

In the unconstrained problem, \(N_{\cM}=\operatorname{span}\{I\}\).
Pauli covariance of the two channels implies that
\(\Omega(I/2)\) is invariant under conjugation by each Pauli operator.
It must therefore be proportional to \(I\).  By
Lemma~\ref{lemma: fixed point of AB}, \(I/2\) is an AB fixed point.
It remains to prove the following inequality for every
\(\rho\in\cS(\cH_A)\) satisfying \(\rho\succ0\) and
\(\rho\ne I/2\):
\begin{equation}
T(\rho):=\Tr[(\rho-I/2)\Omega(\rho)]>0
\label{eq: appendix-C-target}
\end{equation}

\subsection{Radial Representation of \(T(\rho)\)}

Write the Pauli operators as \(X,Y,Z\) and parametrize a full-rank
qubit state distinct from \(I/2\) by
\begin{align}
\rho&=\frac12(I+aX+bY+cZ),\nonumber\\
r&:=\sqrt{a^2+b^2+c^2}\in(0,1),\nonumber\\
u&:=\frac{c^2}{r^2}\in[0,1].
\label{eq: appendix-C-Bloch-parameters}
\end{align}
We use the normalized Bell states
\begin{align}
\ket{\Phi_\pm}
  &:=\frac{1}{\sqrt2}(\ket{00}\pm\ket{11}),&
\ket{\Psi_\pm}
  &:=\frac{1}{\sqrt2}(\ket{01}\pm\ket{10}).
\end{align}
For compactness, let
\begin{align*}
P_{\perp}
&:=\ket{\Phi_-}\bra{\Phi_-}
  +\ket{\Psi_+}\bra{\Psi_+}\\
&\quad+\ket{\Psi_-}\bra{\Psi_-}.
\end{align*}
Then the Choi matrices of the dephasing channel \(\cN_1\) and the
depolarizing channel \(\cN_2\) are
\begin{align}
\Gamma^{\cN_1}
&=2q\ket{\Phi_+}\bra{\Phi_+}
  +2(1-q)\ket{\Phi_-}\bra{\Phi_-},\label{eq: appendix-C-Choi-N1}\\
\Gamma^{\cN_2}
&=2m_p\ket{\Phi_+}\bra{\Phi_+}
  +2\delta_pP_{\perp},
\label{eq: appendix-C-Choi-N2}
\end{align}
where
\begin{equation}
m_p:=1-\frac{3p}{4},\qquad
\delta_p:=\frac p4,\qquad
A_p:=m_p+\delta_p=1-\frac p2.
\label{eq: appendix-C-channel-scalars}
\end{equation}
For later use, define
\begin{align}
R&:=\sqrt{(1-p)^2+p m_p r^2},\label{eq: appendix-C-R}\\
\Delta_u
&:=\sqrt{(2q-1)^2+4q(1-q)r^2u}.
\label{eq: appendix-C-Delta}
\end{align}
For \(p,q\in(0,1)\) and \(r\in(0,1)\), one has
\begin{align}
A_p^2-R^2
&=p m_p(1-r^2)>0,\nonumber\\
1-\Delta_u^2
&=4q(1-q)(1-r^2u)>0.
\label{eq: appendix-C-domains}
\end{align}
Thus \(R/A_p\) and \(\Delta_u\) lie in \((0,1)\) for the parameters
used below.

Set
\begin{align}
A_{\cN_1}(\rho)
&:=(\rho^{1/2}\otimes I_B)\Gamma^{\cN_1}
  (\rho^{1/2}\otimes I_B),\\
A_{\cN_2}(\rho)
&:=(\rho^{1/2}\otimes I_B)\Gamma^{\cN_2}
  (\rho^{1/2}\otimes I_B).
\end{align}
If \(P_{\cN_1}:=\supp A_{\cN_1}\), write
\begin{equation}
\log_{+}A_{\cN_1}
:=P_{\cN_1}\log\!\left(A_{\cN_1}|_{\operatorname{Ran}P_{\cN_1}}\right)P_{\cN_1}.
\end{equation}
Thus all logarithms of \(A_{\cN_1}\) below are taken on its support.
When no confusion can arise, the tensor factor \(I_B\) is suppressed.

\begin{lemma}[Radial identity]
\label{lemma: radial identity for T}
Keep the Bloch direction
\(\widehat n=(a,b,c)/r\) fixed and regard
\(\rho(r)=\frac12(I+r\widehat n\cdot\sigma)\) as a function of \(r\).
Then
\begin{align}
T(\rho)
={}&r\Tr\left[\log A_{\cN_2}(\rho)
              \frac{dA_{\cN_1}(\rho)}{dr}\right]\nonumber\\
&-r\Tr\left[\log_{+}A_{\cN_1}(\rho)
              \frac{dA_{\cN_1}(\rho)}{dr}\right].
\label{eq: appendix-C-radial-identity}
\end{align}
\end{lemma}

\begin{proof}
Let
\begin{equation}
L(\rho):=\log_{+}A_{\cN_1}(\rho)-\log A_{\cN_2}(\rho).
\end{equation}
From \eqref{eq: definition of Omega},
\begin{align}
\Omega_1(\rho)
&=-\Tr_B\!\left[
  \Gamma^{\cN_1}\rho^{1/2}L(\rho)\rho^{-1/2}
  \right],\nonumber\\
\Omega(\rho)
&=\frac{\Omega_1(\rho)+\Omega_1(\rho)^\dagger}{2}.
\end{align}
Because \(\rho-I/2\) is Hermitian,
\begin{equation}
T(\rho)=\operatorname{Re}\Tr[(\rho-I/2)\Omega_1(\rho)].
\end{equation}
In the eigenbasis of \(\rho(r)\),
\begin{equation}
\rho^{-1/2}(\rho-I/2)
=2r\frac{d\rho^{1/2}}{dr}.
\label{eq: appendix-C-square-root-radial}
\end{equation}
Using this identity, cyclicity of the trace, and
\begin{equation}
\frac{dA_{\cN_1}}{dr}
=\frac{d\rho^{1/2}}{dr}\Gamma^{\cN_1}\rho^{1/2}
 +\rho^{1/2}\Gamma^{\cN_1}\frac{d\rho^{1/2}}{dr},
\end{equation}
we obtain
\begin{equation}
T(\rho)
=-r\Tr\left[L(\rho)\frac{dA_{\cN_1}(\rho)}{dr}\right].
\end{equation}
Splitting \(L(\rho)\) into its two logarithmic terms proves
\eqref{eq: appendix-C-radial-identity}.
\end{proof}

\subsection{Closed-Form Expression for \(T(\rho)\)}

The next proposition contains the main calculation.  Its proof follows
the two terms in \eqref{eq: appendix-C-radial-identity} separately.

\begin{proposition}[Closed form of \(T(\rho)\)]
\label{proposition: closed form of T}
For the state in \eqref{eq: appendix-C-Bloch-parameters},
\begin{align}
T(\rho)
={}&
\frac{pq r^2}{2R}\operatorname{artanh}\frac{R}{A_p}\nonumber\\
&+\frac{2(1-q)m_p u r^2}{R}
 \operatorname{artanh}\frac{R}{A_p}\nonumber\\
&+\frac{1-q}{2}(1-u)r\log\frac{1+r}{1-r}\nonumber\\
&-\frac{2q(1-q)u r^2}{\Delta_u}
 \log\frac{1+\Delta_u}{1-\Delta_u}.
\label{eq: close form of Trho}
\end{align}
Moreover, the symmetries of the channel pair imply
\begin{equation}
\Omega(\rho)
=w_0(\rho)I+\eta_{\perp}(\rho)(aX+bY)+\eta_z(\rho)cZ,
\label{eq: appendix-C-Omega-form}
\end{equation}
where
\begin{align}
\eta_{\perp}(\rho)
&=\frac{pq}{2R}\operatorname{artanh}\frac{R}{A_p}
 +\frac{1-q}{2r}\log\frac{1+r}{1-r},
\label{eq: appendix-C-eta-perp}\\
\eta_z(\rho)
&=\frac{pq}{2R}\operatorname{artanh}\frac{R}{A_p}
 +\frac{2(1-q)m_p}{R}\operatorname{artanh}\frac{R}{A_p}\nonumber\\
&\quad-\frac{2q(1-q)}{\Delta_u}
 \log\frac{1+\Delta_u}{1-\Delta_u}.
\label{eq: appendix-C-eta-z}
\end{align}
Consequently,
\begin{equation}
T(\rho)=\eta_{\perp}(\rho)(a^2+b^2)+\eta_z(\rho)c^2.
\label{eq: appendix-C-coefficient-decomposition}
\end{equation}
\end{proposition}

\begin{proof}
We evaluate the two terms in
\eqref{eq: appendix-C-radial-identity}.

\emph{First \(A_{\cN_1}\) term:}
Define
\begin{equation}
\ket{v_\pm}:=(\rho^{1/2}\otimes I_B)\ket{\Phi_\pm}.
\end{equation}
Then
\begin{equation}
A_{\cN_1}
=2q\ket{v_+}\bra{v_+}
 +2(1-q)\ket{v_-}\bra{v_-},
\end{equation}
and
\begin{equation}
\begin{pmatrix}
\braket{v_+|v_+}&\braket{v_+|v_-}\\
\braket{v_-|v_+}&\braket{v_-|v_-}
\end{pmatrix}
=\frac12
\begin{pmatrix}
1&c\\ c&1
\end{pmatrix}.
\end{equation}
The two nonzero eigenvalues of \(A_{\cN_1}\) are therefore
\begin{equation}
\lambda_\pm^{\cN_1}=\frac{1\pm\Delta_u}{2}.
\end{equation}
Because the Bloch direction is fixed, \(u=c^2/r^2\) is constant and
\begin{equation}
\frac{d\lambda_\pm^{\cN_1}}{dr}
=\pm\frac{2q(1-q)ru}{\Delta_u}.
\end{equation}
It follows that
\begin{align}
-r\Tr\left[\log_{+}A_{\cN_1}\frac{dA_{\cN_1}}{dr}\right]
&=-r\sum_{\nu\in\{+,-\}}
  \frac{d\lambda_\nu^{\cN_1}}{dr}\log\lambda_\nu^{\cN_1}\nonumber\\
&=-\frac{2q(1-q)u r^2}{\Delta_u}
  \log\frac{1+\Delta_u}{1-\Delta_u}.
\label{eq: appendix-C-N1-contribution}
\end{align}

\emph{Second term: the \(A_{\cN_2}\) term.}
Both channels are covariant under rotations about the \(Z\)-axis, so
\(T(\rho)\) is invariant under those rotations.  We may therefore set
\(b=0\) and \(a\geq0\).  Put
\begin{equation}
s:=\frac{a}{r}=\sqrt{1-u},
\qquad
t:=\frac{c}{r},
\qquad
s^2+t^2=1.
\end{equation}
Choose a real eigenbasis \(\{\ket{+},\ket{-}\}\) of \(\rho\), with
eigenvalues \(\lambda_\pm=(1\pm r)/2\), and define
\begin{align}
e_+&:=\ket{+}\ket{+},&
e_-&:=\ket{-}\ket{-},\nonumber\\
f_+&:=\ket{+}\ket{-},&
f_-&:=\ket{-}\ket{+}.
\end{align}
In this basis,
\begin{align}
\ket{v_+}
&=\frac{\sqrt{1+r}}{2}e_+
 +\frac{\sqrt{1-r}}{2}e_-,
\label{eq: appendix-C-v-plus}\\
\ket{v_-}
&=\frac{t}{2}\bigl(\sqrt{1+r}\,e_+
                  -\sqrt{1-r}\,e_-\bigr)\nonumber\\
&\quad-\frac{s}{2}\bigl(\sqrt{1+r}\,f_+
                       +\sqrt{1-r}\,f_-\bigr).
\label{eq: appendix-C-v-minus}
\end{align}
Using
\begin{equation}
\Gamma^{\cN_2}
=2\delta_p I_{AB}
 +2(1-p)\ket{\Phi_+}\bra{\Phi_+},
\end{equation}
we see that \(A_{\cN_2}\) is block diagonal in the decomposition
\(\operatorname{span}\{e_+,e_-\}\oplus
\operatorname{span}\{f_+\}\oplus\operatorname{span}\{f_-\}\).
On the first block it is represented by
\begin{equation}
B_e=\frac12
\begin{pmatrix}
A_p(1+r)&(1-p)\sqrt{1-r^2}\\
(1-p)\sqrt{1-r^2}&A_p(1-r)
\end{pmatrix},
\label{eq: appendix-C-Be}
\end{equation}
whose eigenvalues are \((A_p\pm R)/2\).  Hence
\begin{equation}
\log B_e
=C_0I+\frac{2}{R}\operatorname{artanh}\frac{R}{A_p}
       \left(B_e-\frac{A_p}{2}I\right),
\label{eq: appendix-C-log-Be}
\end{equation}
where \(C_0\) is a scalar.  The remaining two eigenvalues of
\(A_{\cN_2}\) are \(\delta_p(1+r)\) and \(\delta_p(1-r)\), on
\(f_+\) and \(f_-\), respectively.

The norms of the \(e_\pm\)-block components of \(v_+\) and \(v_-\)
are independent of \(r\), so the scalar term \(C_0I\) in
\eqref{eq: appendix-C-log-Be} does not contribute to the following
derivatives.  Substituting
\eqref{eq: appendix-C-v-plus}--\eqref{eq: appendix-C-log-Be} gives
\begin{align}
\left\langle\frac{dv_+}{dr}\middle|
 \log A_{\cN_2}\middle|v_+\right\rangle
&=\frac{pr}{8R}\operatorname{artanh}\frac{R}{A_p},
\label{eq: appendix-C-v-plus-element}\\
\left\langle\frac{dv_-}{dr}\middle|
 \log A_{\cN_2}\middle|v_-\right\rangle
&=\frac{m_p u r}{2R}\operatorname{artanh}\frac{R}{A_p}\notag\\
&\quad+\frac{1-u}{8}\log\frac{1+r}{1-r}.
\label{eq: appendix-C-v-minus-element}
\end{align}
Here the first term in \eqref{eq: appendix-C-v-minus-element} comes
from the \(e_\pm\) block, while the logarithmic term comes from the
one-dimensional \(f_\pm\) blocks.  Since
\begin{align}
\Tr\left[\log A_{\cN_2}\frac{dA_{\cN_1}}{dr}\right]
={}&4q\left\langle\frac{dv_+}{dr}\middle|
       \log A_{\cN_2}\middle|v_+\right\rangle\nonumber\\
&+4(1-q)\left\langle\frac{dv_-}{dr}\middle|
       \log A_{\cN_2}\middle|v_-\right\rangle,
\end{align}
we obtain
\begin{align}
r\Tr\left[\log A_{\cN_2}\frac{dA_{\cN_1}}{dr}\right]
={}&\frac{pq r^2}{2R}\operatorname{artanh}\frac{R}{A_p}\nonumber\\
&+\frac{2(1-q)m_p u r^2}{R}
 \operatorname{artanh}\frac{R}{A_p}\nonumber\\
&+\frac{1-q}{2}(1-u)r\log\frac{1+r}{1-r}.
\label{eq: appendix-C-N2-contribution}
\end{align}
Combining \eqref{eq: appendix-C-N1-contribution} and
\eqref{eq: appendix-C-N2-contribution} proves
\eqref{eq: close form of Trho}.

Rotational and reflection symmetries about the \(Z\)-axis force the
Bloch components of \(\Omega(\rho)\) to have the form
\eqref{eq: appendix-C-Omega-form}.  Since
\(a^2+b^2=(1-u)r^2\) and \(c^2=ur^2\), comparison with
\eqref{eq: close form of Trho} gives
\eqref{eq: appendix-C-eta-perp}--\eqref{eq: appendix-C-coefficient-decomposition}.
\end{proof}

\subsection{Strict Positivity of \(T(\rho)\)}

\begin{lemma}[Positivity of the Bloch coefficients]
\label{lemma: positivity of eta coefficients}
For
\begin{equation}
p=\frac1{20},\qquad q=\frac25,
\end{equation}
the coefficients in \eqref{eq: appendix-C-eta-perp} and
\eqref{eq: appendix-C-eta-z} satisfy
\begin{equation}
\eta_{\perp}(\rho)>0,\qquad \eta_z(\rho)>0
\end{equation}
for every \(r\in(0,1)\) and \(u\in[0,1]\).
\end{lemma}

\begin{proof}
Every term in \eqref{eq: appendix-C-eta-perp} is positive, so
\(\eta_{\perp}(\rho)>0\).  For \(\eta_z\), the first term in
\eqref{eq: appendix-C-eta-z} is positive.  It remains to show that
the second term dominates the magnitude of the third, or equivalently,
\begin{equation}
\frac{m_p}{R}\operatorname{artanh}\frac{R}{A_p}
>
\frac{2q}{\Delta_u}\operatorname{artanh}\Delta_u.
\label{eq: appendix-C-eta-z-target}
\end{equation}
Define
\begin{equation}
h(x):=\frac{\operatorname{artanh}x}{x},
\qquad 0<x<1.
\end{equation}
The power series
\begin{equation}
h(x)=\sum_{k=0}^{\infty}\frac{x^{2k}}{2k+1}
\end{equation}
shows that \(h\) is strictly increasing on \((0,1)\).

For the chosen parameters,
\begin{equation}
m_p=\frac{77}{80},\qquad A_p=\frac{39}{40},
\end{equation}
and, because \(u\leq1\),
\begin{align}
\Delta_u^2
&\leq\Delta_1^2:=\frac{1+24r^2}{25},\\
\left(\frac{R}{A_p}\right)^2
&=\frac{1444+77r^2}{1521}.
\end{align}
Their difference is
\begin{equation}
\left(\frac{R}{A_p}\right)^2-\Delta_1^2
=\frac{34579}{38025}(1-r^2)>0.
\end{equation}
Thus
\begin{equation}
\frac{R}{A_p}>\Delta_1\geq\Delta_u.
\end{equation}
Monotonicity of \(h\), together with
\begin{equation}
\frac{m_p}{A_p}=\frac{77}{78}>\frac45=2q,
\end{equation}
gives
\begin{equation}
\frac{m_p}{A_p}h\!\left(\frac{R}{A_p}\right)
>2q\,h(\Delta_u),
\end{equation}
which is precisely \eqref{eq: appendix-C-eta-z-target}.  Hence
\(\eta_z(\rho)>0\).
\end{proof}

\begin{proof}[Proof of Theorem~\ref{theorem: convergence without constraint}]
For every \(\rho\neq I/2\), at least one of \(a,b,c\) is nonzero.
Proposition~\ref{proposition: closed form of T} and
Lemma~\ref{lemma: positivity of eta coefficients} therefore give
\begin{equation}
T(\rho)
=\eta_{\perp}(\rho)(a^2+b^2)+\eta_z(\rho)c^2>0.
\end{equation}
We already showed that \(I/2\) is an AB fixed point.  Applying
Lemma~\ref{lemma: scalar fixed-point certificate} with
\(\widehat\rho=I/2\) proves that it is the unique full-rank fixed point.
\end{proof}

\section{Proof of Theorem~\ref{theorem: convergence with constraint}}
\label{proof of theorem convergence with constraint}

Fix \(H\in\{X,Y,Z\}\).  We prove uniqueness for the constrained
problem stated in Theorem~\ref{theorem: convergence with constraint}.
As in the rest of the paper, fixed points are considered in
\(\cM_{++}\); in particular, \(-1<E<1\).  No new spectral calculation
is required.  We reuse the symmetry decomposition and the strictly
positive coefficients established in
Appendix~\ref{proof of theorem convergence without constraint}.

Let
\begin{equation}
\widehat\rho_H:=\frac12(I+EH).
\end{equation}

\begin{proof}[Proof of Theorem~\ref{theorem: convergence with constraint}]
We first verify that \(\widehat\rho_H\) is an AB fixed point.  If
\(E=0\), then \(\widehat\rho_H=I/2\), whose fixed-point property was
established in
Appendix~\ref{proof of theorem convergence without constraint}.
Suppose now that \(E\ne0\).  At \(\widehat\rho_X\), the symmetry form
\eqref{eq: appendix-C-Omega-form} gives
\(\Omega(\widehat\rho_X)\in\operatorname{span}\{I,X\}\).  Similarly,
\begin{align}
\Omega(\widehat\rho_Y)&\in\operatorname{span}\{I,Y\},&
\Omega(\widehat\rho_Z)&\in\operatorname{span}\{I,Z\}.
\end{align}
Thus, in every case,
\(\Omega(\widehat\rho_H)\in N_{\cM}\), and
Lemma~\ref{lemma: fixed point of AB} shows that
\(\widehat\rho_H\) is an AB fixed point.

It remains to establish the strict positivity required by
Lemma~\ref{lemma: scalar fixed-point certificate}.  We use the
representation \eqref{eq: appendix-C-Omega-form} and the positivity
proved in Lemma~\ref{lemma: positivity of eta coefficients}.

If \(H=X\), every feasible state can be written as
\begin{equation}
\rho=\frac12(I+EX+bY+cZ).
\end{equation}
Orthogonality of the Pauli operators gives
\begin{equation}
\Tr[(\rho-\widehat\rho_X)\Omega(\rho)]
=b^2\eta_{\perp}(\rho)+c^2\eta_z(\rho).
\end{equation}
This is strictly positive unless \(b=c=0\), which is exactly the
excluded state \(\widehat\rho_X\).

If \(H=Y\), write
\begin{equation}
\rho=\frac12(I+aX+EY+cZ).
\end{equation}
The same calculation gives
\begin{equation}
\Tr[(\rho-\widehat\rho_Y)\Omega(\rho)]
=a^2\eta_{\perp}(\rho)+c^2\eta_z(\rho)>0
\end{equation}
away from \(\widehat\rho_Y\).

Finally, if \(H=Z\), write
\begin{equation}
\rho=\frac12(I+aX+bY+EZ).
\end{equation}
Then
\begin{equation}
\Tr[(\rho-\widehat\rho_Z)\Omega(\rho)]
=(a^2+b^2)\eta_{\perp}(\rho)>0
\end{equation}
away from \(\widehat\rho_Z\).  Therefore, for each
\(H\in\{X,Y,Z\}\),
\begin{equation}
\Tr[(\rho-\widehat\rho_H)\Omega(\rho)]>0
\qquad
\text{for every }
\rho\in\cM_{++}\setminus\{\widehat\rho_H\}.
\end{equation}
Lemma~\ref{lemma: scalar fixed-point certificate} now proves that
\(\widehat\rho_H\) is the unique fixed point in \(\cM_{++}\).
\end{proof}
\section{Gradient of the Channel-Relative-Entropy Objective}
\label{appendix: calculation of derivative}

This appendix derives the explicit expression for $\nabla f(\rho)$ used in Section~\ref{sec: examples} to compare the AB and MD update directions. Direct differentiation of $f$ is inconvenient
because $\rho$ occurs both under a square root and inside matrix logarithms. We therefore proceed in three stages: we reparametrize $\rho=X^2$, differentiate the resulting function with respect to $X$,
and then recover $\nabla f(\rho)$ from a Sylvester equation.

\subsection{Setting and Proof Strategy}

Let $\rho\succ0$ be a state on $A$, set
\begin{equation}
X:=\rho^{1/2},\qquad \widetilde X:=X\otimes I_B,
\end{equation}
and define
\begin{equation}
A_{\cN_1}:=\widetilde X\Gamma^{\cN_1}\widetilde X,
\qquad
A_{\cN_2}:=\widetilde X\Gamma^{\cN_2}\widetilde X.
\label{eq:appendix-output-operators}
\end{equation}
In the example considered in this paper, $\Gamma^{\cN_1}\succeq0$ and
$\Gamma^{\cN_2}\succ0$. Hence $A_{\cN_2}\succ0$, whereas $A_{\cN_1}$
may be singular. For any positive semidefinite $A$, we use the
support logarithm
\begin{equation}
\log_{+}A:=P\log\!\left(A|_{\operatorname{Ran}P}\right)P,
\qquad P:=\supp A.
\label{eq:appendix-support-log}
\end{equation}
With the convention $0\log 0=0$, the objective in
\eqref{eq: definition of f} becomes
\begin{align}
f(\rho)
&=-\Tr_{AB}\!\left[A_{\cN_1}
  \left(\log_{+}A_{\cN_1}-\log A_{\cN_2}\right)\right].
\label{eq:appendix-f}
\end{align}
Define $g(X):=f(X^2)$. We first calculate $\nabla g(X)$ and then use
the chain rule in the following lemma.

\begin{lemma}[Square-root chain rule]
\label{lem:appendix-square-root-chain-rule}
Let the gradients be defined by
\begin{align}
Df(\rho)[\delta\rho]
  &=\Tr_A[\nabla f(\rho)\delta\rho],\\
Dg(X)[\delta X]
  &=\Tr_A[\nabla g(X)\delta X]
\end{align}
for Hermitian perturbations. Then
\begin{equation}
\nabla g(X)=X\nabla f(\rho)+\nabla f(\rho)X.
\label{eq:appendix-chain-rule}
\end{equation}
Consequently, $\nabla f(\rho)$ is uniquely determined by
$\nabla g(X)$ whenever $X\succ0$.
\end{lemma}

\begin{proof}
The variation of $\rho=X^2$ is
$\delta\rho=X\delta X+\delta X X$. Therefore,
\begin{align}
Dg(X)[\delta X]
&=Df(\rho)[X\delta X+\delta X X]\nonumber\\
&=\Tr_A\!\left[
  \bigl(\nabla f(\rho)X+X\nabla f(\rho)\bigr)\delta X
  \right],
\end{align}
which proves \eqref{eq:appendix-chain-rule}. Since every eigenvalue
of $X$ is positive, the map $Z\mapsto XZ+ZX$ is invertible on the
Hermitian matrices.
\end{proof}

\subsection{Two Matrix-Differential Identities}

The first identity handles the possibly singular term
$\Tr(A_{\cN_1}\log A_{\cN_1})$. Its hypothesis is satisfied here because
$\widetilde X$ is invertible and thus
\begin{equation}
\operatorname{rank}A_{\cN_1}
=\operatorname{rank}\!\left(\widetilde X\Gamma^{\cN_1}\widetilde X\right)
=\operatorname{rank}\Gamma^{\cN_1}
\label{eq:appendix-constant-rank}
\end{equation}
throughout the positive-definite domain of $X$.

\begin{lemma}[Entropy derivative at constant rank]
\label{lem:constant-rank-entropy}
Let $A(t)\succeq0$ be differentiable and have constant rank near
$t=0$. Write $A:=A(0)$, $\dot A:=A'(0)$, and $P:=\supp A$. Then
\begin{equation}
\left.\frac{d}{dt}\Tr[A(t)\log A(t)]\right|_{t=0}
=\Tr[(\log_{+}A+P)\dot A].
\label{eq:appendix-entropy-derivative}
\end{equation}
\end{lemma}

\begin{proof}
Let $r$ be the common rank and choose a differentiable local
factorization $A(t)=Y(t)Y(t)^{\dagger}$, where $Y(t)$ has $r$
linearly independent columns. The positive-definite matrix
$H(t):=Y(t)^{\dagger}Y(t)$ has the same nonzero eigenvalues as $A(t)$.
Hence
\begin{equation}
\Tr[A(t)\log A(t)]=\Tr[H(t)\log H(t)].
\end{equation}
At $t=0$, write the polar decomposition as $Y=UH^{1/2}$. Then
$P=UU^{\dagger}$, $A=UHU^{\dagger}$, and
$\log_{+}A=U(\log H)U^{\dagger}$. Using
\begin{align}
\dot A&=\dot Y Y^{\dagger}+Y\dot Y^{\dagger},&
\dot H&=\dot Y^{\dagger}Y+Y^{\dagger}\dot Y,
\end{align}
cyclicity of the trace gives
\begin{align}
\left.\frac{d}{dt}\Tr[H(t)\log H(t)]\right|_{t=0}
&=\Tr[(\log H+I)\dot H]\nonumber\\
&=\Tr[(\log_{+}A+P)\dot A],
\end{align}
which proves the claim.
\end{proof}

The second identity treats the derivative of $\log A_{\cN_2}$ in the
cross term $\Tr(A_{\cN_1}\log A_{\cN_2})$.

\begin{lemma}[Fr\'echet derivative of the logarithm]
\label{lem:log-frechet}
For $A\succ0$, define the linear map
\begin{equation}
\mathcal T_A(E):=D(\log A)[E]
=\int_0^\infty(A+sI)^{-1}E(A+sI)^{-1}\,ds.
\label{eq:appendix-log-frechet}
\end{equation}
This map is self-adjoint with respect to the Hilbert--Schmidt inner
product:
\begin{equation}
\Tr[C\mathcal T_A(E)]=\Tr[\mathcal T_A(C)E].
\label{eq:appendix-log-self-adjoint}
\end{equation}
If $A=U\operatorname{diag}(m_1,\ldots,m_n)U^{\dagger}$, where
$m_i>0$, then
\begin{equation}
\mathcal T_A(E)=U\bigl(K\odot(U^{\dagger}EU)\bigr)U^{\dagger},
\label{eq:appendix-log-spectral}
\end{equation}
where $\odot$ denotes the Hadamard product and
\begin{equation}
K_{ij}=\begin{cases}
\dfrac{\log m_i-\log m_j}{m_i-m_j},&m_i\ne m_j,\\[2mm]
\dfrac{1}{m_i},&m_i=m_j.
\end{cases}
\end{equation}
\end{lemma}

\begin{proof}
The integral representation in \eqref{eq:appendix-log-frechet}
immediately implies \eqref{eq:appendix-log-self-adjoint} by cyclicity
of the trace. Evaluating the integral in an eigenbasis of $A$ gives
\eqref{eq:appendix-log-spectral}.
\end{proof}

\subsection{Gradient Formula}

We now state and prove the desired expression. Besides giving the
gradient, the proposition separates the contribution obtained by
differentiating $A_{\cN_1}$ from the contribution obtained by
differentiating $\log A_{\cN_2}$.

\begin{proposition}[Gradient of $f$]
\label{proposition:appendix-gradient-formula}
Let
\begin{align}
P_{\cN_1}&:=\supp A_{\cN_1},\nonumber\\
G&:=\log_{+}A_{\cN_1}+P_{\cN_1}-\log A_{\cN_2},\nonumber\\
L&:=\mathcal T_{A_{\cN_2}}(A_{\cN_1}).
\label{eq:appendix-G-L}
\end{align}
Then
\begin{align}
\nabla g(X)=-\Tr_B\bigl(&
\Gamma^{\cN_1}\widetilde XG+G\widetilde X\Gamma^{\cN_1}
\nonumber\\[-1mm]
&-\Gamma^{\cN_2}\widetilde XL-L\widetilde X\Gamma^{\cN_2}
\bigr).
\label{eq:appendix-gradient-g}
\end{align}
Moreover, $\nabla f(\rho)$ is the unique Hermitian solution of
\begin{equation}
X\nabla f(\rho)+\nabla f(\rho)X=\nabla g(X).
\label{eq:appendix-gradient-f-sylvester}
\end{equation}
Equivalently, if
$X=U\operatorname{diag}(x_1,\ldots,x_d)U^{\dagger}$, then
\begin{equation}
\bigl[U^{\dagger}\nabla f(\rho)U\bigr]_{ij}
=\frac{\bigl[U^{\dagger}\nabla g(X)U\bigr]_{ij}}{x_i+x_j}.
\label{eq:appendix-gradient-f-eigenbasis}
\end{equation}
\end{proposition}

\begin{proof}
For a Hermitian perturbation $\delta X$, set
$\delta\widetilde X:=(\delta X)\otimes I_B$. The corresponding
variations of the two output operators are
\begin{align}
\delta A_{\cN_1}
&=(\delta\widetilde X)\Gamma^{\cN_1}\widetilde X
  +\widetilde X\Gamma^{\cN_1}(\delta\widetilde X),
\label{eq:appendix-delta-N1}\\
\delta A_{\cN_2}
&=(\delta\widetilde X)\Gamma^{\cN_2}\widetilde X
  +\widetilde X\Gamma^{\cN_2}(\delta\widetilde X).
\label{eq:appendix-delta-N2}
\end{align}

By \eqref{eq:appendix-constant-rank} and
Lemma~\ref{lem:constant-rank-entropy},
\begin{equation}
D\Tr[A_{\cN_1}\log A_{\cN_1}][\delta X]
=\Tr[(\log_{+}A_{\cN_1}+P_{\cN_1})\delta A_{\cN_1}].
\label{eq:appendix-first-trace-term}
\end{equation}
For the cross term, the product rule gives
\begin{align}
D\Tr[A_{\cN_1}\log A_{\cN_2}][\delta X]
={}&\Tr[(\delta A_{\cN_1})\log A_{\cN_2}]\nonumber\\
&+\Tr[A_{\cN_1}\mathcal T_{A_{\cN_2}}(\delta A_{\cN_2})].
\label{eq:appendix-cross-term}
\end{align}
The self-adjointness in \eqref{eq:appendix-log-self-adjoint} and the
definition of $L$ imply
\begin{equation}
\Tr[A_{\cN_1}\mathcal T_{A_{\cN_2}}(\delta A_{\cN_2})]
=\Tr[L\delta A_{\cN_2}].
\label{eq:appendix-cross-self-adjoint}
\end{equation}
Combining \eqref{eq:appendix-first-trace-term}--
\eqref{eq:appendix-cross-self-adjoint} with the minus sign in
\eqref{eq:appendix-f}, we obtain
\begin{equation}
Dg(X)[\delta X]
=-\Tr[G\delta A_{\cN_1}]+\Tr[L\delta A_{\cN_2}].
\label{eq:appendix-Dg-output-variations}
\end{equation}

Substitute \eqref{eq:appendix-delta-N1} and
\eqref{eq:appendix-delta-N2} into
\eqref{eq:appendix-Dg-output-variations}, and move
$\delta\widetilde X$ to the left by cyclicity of the trace. This
yields
\begin{align}
Dg(X)[\delta X]= -\Tr_{AB}\Bigl[(\delta\widetilde X)\bigl(&
\Gamma^{\cN_1}\widetilde XG+G\widetilde X\Gamma^{\cN_1}
\nonumber\\[-1mm]
&-\Gamma^{\cN_2}\widetilde XL-L\widetilde X\Gamma^{\cN_2}
\bigr)\Bigr].
\label{eq:appendix-Dg-delta-X}
\end{align}
For every operator $Y$ on $AB$,
\begin{equation}
\Tr_{AB}[((\delta X)\otimes I_B)Y]
=\Tr_A[\delta X\Tr_B(Y)].
\label{eq:appendix-partial-trace-identity}
\end{equation}
Comparing \eqref{eq:appendix-Dg-delta-X} with the definition of
$\nabla g(X)$ proves \eqref{eq:appendix-gradient-g}.

Equation \eqref{eq:appendix-gradient-f-sylvester} follows from
Lemma~\ref{lem:appendix-square-root-chain-rule}. In an eigenbasis of
$X$, its $(i,j)$ entry is
\begin{equation}
(x_i+x_j)\bigl[U^{\dagger}\nabla f(\rho)U\bigr]_{ij}
=\bigl[U^{\dagger}\nabla g(X)U\bigr]_{ij}.
\end{equation}
Because $x_i+x_j>0$, division gives
\eqref{eq:appendix-gradient-f-eigenbasis} and also proves uniqueness.
\end{proof}

Thus, the gradient used in the numerical comparisons is obtained by
computing $L$ from \eqref{eq:appendix-log-spectral}, evaluating
\eqref{eq:appendix-gradient-g}, and solving the Sylvester equation
\eqref{eq:appendix-gradient-f-sylvester}.

\bibliographystyle{IEEEtran}
\bibliography{references}
\end{document}